\documentclass[12pt,showkeys,nofootinbib,longbibliography]{iopart}
\usepackage{iopams} % Math expressions according to IOP standard
\usepackage{amsthm}
\usepackage{amstext}
\usepackage[hang,small,bf]{caption} % For figure captions
\usepackage{graphicx} % For figure file referencing
\usepackage{mathdots} % For "..." in the matrices
\usepackage{epstopdf} % For inserting .EPS figures

\newtheorem{mydef}{Definition}
\newtheorem{lem}{Lemma}
\newtheorem{prop}{Proposition}
\newtheorem{theorem}{Theorem}

\newcommand{\reals}{{\mathbb R}}
\newcommand{\la}{{\lambda}}

\newcommand{\Braket}[2]{\left<#1|#2\right>}
\newcommand{\Ketbra}[2]{\ket{#1}\!\bra{#2}}

\input{Qcircuit.tex}  % Necessary for drawing quantum circuits

\begin{document}

\title[Quantum Poisson solver]{Quantum algorithm and circuit design solving the Poisson equation}

\author{Yudong Cao$^1$, Anargyros Papageorgiou$^3$, Iasonas Petras$^3$, Joseph Traub$^3$ and Sabre Kais$^{2}$\footnote{Corresponding author. Email: kais@purdue.edu}}

\address{$^1$ Department of Mechanical Engineering, Purdue University, West Lafayette, IN 47907}
\address{$^2$ Department of Chemistry, Physics, and Computer Science; Birck Nanotechnology Center, Purdue University, West Lafayette, IN 47907}
%\ead{kais@purdue.edu}
\address{$^3$ Department of Computer Science, Columbia University, New York, 10027}

\begin{abstract}
The Poisson equation occurs in many areas of science and engineering. Here we
focus on its numerical solution for an equation in $d$ dimensions. 
In particular we present a quantum algorithm and a scalable quantum circuit design which approximates 
the solution of the Poisson equation on a grid with error $\varepsilon$. 
We assume we are given a supersposition of function evaluations of the right hand side of the Poisson equation.
The algorithm produces a quantum state encoding the solution. 
The number of quantum operations and the number of qubits used by the circuit is almost linear in $d$ and polylog in $\varepsilon^{-1}$.
%The circuit uses a number of qubits which is also almost linear in $d$ and polylog in $\varepsilon^{-1}$. 
We present quantum circuit modules together with performance 
guarantees which can be also used for other problems. 
\end{abstract}
%\submitto{\NJP}

\section{Introduction}
Quantum computers take advantage of quantum mechanics to solve certain computational problems faster than
classical computers. Indeed in some cases the quantum algorithm is exponentially faster than the best classical algorithm known
\cite{abrams97,abrams99,alan05,dowling06,lidar99,lloyd96,traub10,shor94,zhang10,nori11a,kais08,nori11b}.
%It is important to distinguish between two notions of \lq\lq exponentially faster\rq\rq. We introduce a new dichotomy:

%\begin{enumerate}
%\item Exponential quantum speedup: A quantum computer can solve the problem exponentially faster than any {\em known} classical algorithm.
%\item Strong exponential quantum speedup: A quantum computer can solve the problem exponentially faster than any classical algorithm.
%\end{enumerate}

%The crucial difference between the two concepts is that if a problem satisfies criterion 1 someone may invent a faster classical
%algorithm, while if a problem satisfies criterion 2 no faster classical algorithm can exist.

%To prove {\em strong} exponential quantum speedup good bounds on the classical computational complexity must be known.
%Shor's very influential result \cite{shor94} yields exponential speedup but not strong exponential speedup because the
%classical computational complexity of factorization is not known. The same is true of most of the quantum speedups for scientific
%problems. Of course, exponential can be replaced by any other function. Thus Grover's algorithm for 
%searching an unstructured database \cite{grover97} enjoys strong polynomial quantum speedup. 

In this paper we present a quantum algorithm and circuit solving the Poisson equation.
The Poisson equation plays a fundamental role in numerous areas of science and engineering, such as computational fluid 
dynamics \cite{b00,fletcher91}, quantum mechanical continuum solvation \cite{tmc05}, electrostatics \cite{griffiths99}, the theory of Markov chains \cite{meyn07,meyn09,asmussen07} 
and is important for density functional theory and electronic structure calculations \cite{engel11}. 
% In particular, we consider solving the discretized Poisson equation $-\Delta_h\vec{v}=\vec{f}$ with Dirichlet boundary condition where $\Delta_h$ is the Poisson matrix, $\vec{v}$ is the solution of the equation on a discretized domain and $\vec{f}$ is the value of the inhomogeneous term at the corresponding grid points in the domain (For details see Section~\ref{sec:DP}).

Any classical numerical algorithm solving the Poisson equation with error $\varepsilon$ 
has cost bounded from below by a function that
grows as $\varepsilon^{-\alpha d}$, where $d$ denotes the 
dimension or 
the number of variables, and $\alpha > 0$ is a smoothness constant \cite{wer91,rw96}. 
Therefore the cost grows exponentially in
$d$ and the problem suffers from the curse of dimensionality. 

We show that the Poisson equation can be solved with error $\varepsilon$ 
using a quantum algorithm with a number of quantum operations which is almost linear in $d$ 
and polylog in $\varepsilon^{-1}$. A number of repetitions 
proportional to $\varepsilon^{-4\alpha}$ guarantees that this algorithm 
succeeds with probability arbitrarily close to $1$. Hence
the quantum algorithm breaks the curse of dimensionality
and, with respect to the dimension of the problem $d$,
enjoys exponential speedup relative to classical algorithms.

On the other hand, we point out that the output of the algorithm is a quantum state that encodes the solution on a regular grid rather than a bit string that represents the solution. 
It can be useful if one is interested in computing a function of the solution rather
than the solution itself. In general, the quantum circuit implementing 
the algorithm can be used as a module
in other quantum algorithms that need the solution of the Poisson equation
to achieve their main task.

In terms of the input of the algorithm, we assume that a quantum state encoding a superposition of
function evaluations of the
right hand side of the Poisson equation is available to us, and we do not account for the cost for
preparing this superposition. In general, preparing
arbitrary quantum states is a very hard problem. Nevertheless, in certain cases one 
can prepare efficiently superpositions of function evaluations using techniques in \cite{grover02,SS06}.
We do not deal with the implementation of such superpositions in this paper.

%We will show that for the Poisson equation we achieve strong exponential quantum speedup for computing $\vec{v}^TM\vec{v}$ for some operator $M$ which can be mapped quantum mechanically. It is known that the worst case classical complexity
%of the Poisson equation is exponential in the dimension $d$ \cite{wer91}. Using the quantum algorithm, we obtain the solution to the discretized equation which is encoded in a quantum state $|v\rangle=\sum_iv_i|i\rangle$ (here $v_i$ is the $i$-th element of $\vec{v}$) with error $\varepsilon$ with
%a number of quantum operations which is almost linear in $d$ and polylog in $\varepsilon^{-1}$. 
%More precisely, the number of quantum operations is
%proportional to $\max\{d, \log_2\varepsilon^{-1}\}(\log_2 d+ \log_2\varepsilon^{-1})^3$. The quantum circuit uses a number of qubits proportional to 
%$\max\{d, \log_2\varepsilon^{-1}\} (\log_2 d + \log_2\varepsilon^{-1})^2$. 

There are many ways to solve the Poisson equation. We choose to discretize it on a regular grid in cartesian coordinates and then solve the resulting system of
linear equations. For this we use the quantum algorithm of \cite{lloyd09} for solving 
systems of linear equation. The solution of differential and partial differential equations
is a natural candidate for applying that algorithm, as already stated in \cite{lloyd09}. 
It has been applied to the solution of differential equations in~\cite{berry10,LO08}. In the case of the Poisson equation, however, that we consider in this paper there is no 
need to assume that the matrix
is given by an oracle. Indeed, a significant part of our work 
deals with the Hamiltonian simulation of the matrix of the
Poisson equation. Moreover, it is an open problem to determine when it is possible to simulate a Hamiltonian with cost polynomial in 
the logarithm of the matrix size and the logarithm of $\varepsilon^{-1}$ \cite{CW12}. Our results show that in the case of the Hamiltonian for the Poisson equation the answer is positive.

%Thus we can use the quantum algorithm of Harrow et al. \cite{lloyd09} for solving the system without assuming, however,
%that the matrix is given by an oracle. A similar idea of using the algorithm~\cite{lloyd09} for solving differential equations is presented in~\cite{berry10}. However, in this paper, a significant part of our work deals with the Hamiltonian simulation of the matrix of the
%Poisson equation. Moreover, it is an open problem to determine whether or not it is possible to simulate a Hamiltonian with cost polynomial in 
%the logarithm of the matrix size and the logarithm of $\varepsilon^{-1}$ \cite{CW12}. Our results show %that in the case of the Hamiltonian in the Poisson equation the answer is positive.  

Our analysis of the implementation includes all the numerical details and will be helpful to researchers working on other
problems. 
All calculations are carried out in fixed precision arithmetic and we provide
accuracy and cost guarantees.
We account for the qubits, including
ancilla qubits, needed for the different operations. We provide quantum circuit modules for the approximation of trigonometric functions, which are needed
in the Hamiltonian simulation of the matrix of the Poisson equation. 
We show how to obtain a quantum circuit computing the reciprocal of the eigenvalues using Newton iteration and modular addition and 
multiplication. 
We show how to implement quantum mechanically the inverse trigonometric function 
needed for controlled rotations. 
As we indicated, our results are not limited to the solution of the Poisson equation but can be used in other quantum algorithms.
Our simulation module  
can be combined with splitting methods to simulate the Hamiltonian $-\Delta + V$,
where $\Delta$ is the Laplacian and  $V$ is a potential function. 
The trigonometric approximations can be used by algorithms dealing with quantum walks. 
The reciprocal of a real number and a controlled rotation by an angle obtained by an inverse trigonometric approximation
are needed for implementing the linear systems algorithm \cite{lloyd09} regardless of the matrix involved.

%In developing the quantum circuit we need certain modules along with performance guarantees which we believe will be helpful 
%to researchers working on other problems. The modules include reciprocals of eigenvalues and trigonometric approximations.

\section{Overview}

We consider the $d$-dimensional Poisson equation with Dirichlet boundary conditions.

\begin{mydef}
\begin{eqnarray}
-\Delta u(x) = f(x) & \quad & x\in I_d:=(0,1)^d,  \label{eq:poisson}\\
u(x) =0 & \quad & x\in \partial I_d \nonumber, 
\end{eqnarray}
where $f:I_d\rightarrow\reals$ is a sufficiently smooth function; e.g., see \cite{evans98,forth04,wer91} for details.
\end{mydef}

For simplicity we study this equation over the unit cube but a similar analysis applies to more general domains in
$\reals^d$. Often one solves this equation by discretizing it and solving the resulting 
linear system. 
A finite difference discretization of the Poisson equation on a grid with mesh size $h$, using a $(2d+1)$ stencil for the Laplacian,
yields the linear system

\begin{equation}
\label{eq:SYSTEM}
{\text{$-\Delta_h \vec v = \vec f_h,$}}
\end{equation}
where $f_h$ is the vector obtained by sampling the function $f$ on the interior grid points \cite{demmel97,forth04,lev07}.
The resulting matrix is symmetric positive definite. 

To solve the Poisson equation with error $O(\varepsilon)$ both the discretization error and the error 
on the solution of the system should be $O(\varepsilon)$.
This implies that $\Delta_h$ is a 
matrix of size proportional to $\varepsilon^{-{\alpha}d} \times \varepsilon^{-{\alpha}d}$, 
where ${\alpha}>0$ is a constant that depends on the smoothness of the solution which, in turn, depends on the smoothness of $f$ \cite{BH62,forth04,wer91}.
For example, when $f$ has uniformly bounded partial derivatives up to order four then $\alpha=1/2$. 

There are different ways for solving this system using classical algorithms. Demmel \cite[Table 6.1]{demmel97} lists a number of possibilities.
The conjugate gradient algorithm \cite{saad03} is an example.
Its cost for solving this system with error $\varepsilon$ is proportional to 
\[\varepsilon^{-\alpha d} \sqrt{\kappa} \log\varepsilon^{-1},\]
where $\kappa$ denotes the condition number of $\Delta_h$. We know $\kappa=\varepsilon^{-2\alpha}$, 
independently of $d$.
The resulting cost is proportional to  
$\varepsilon^{-{\alpha}d-\alpha} \log\varepsilon^{-1}$. For details about the solution of large linear systems see \cite{TW84}.
Observe that the factor $\varepsilon^{-{\alpha}d}$ in the cost is the matrix size and its contribution cannot be overcome.
Any direct or iterative 
classical algorithm solving this system has cost at least 
$\varepsilon^{-{\alpha}d}$, since the algorithm must determine all unknowns. So any algorithm solving the system has cost exponential in $d$. 
In fact a much stronger result holds, namely, the cost of any classical algorithm solving the Poisson equation in the worst case
must be exponential in $d$ \cite{wer91}.

We present a scalable quantum circuit for the solution of (\ref{eq:SYSTEM})
and thereby for the solution of Poisson equation with error $O(\varepsilon)$ that 
uses a number of qubits proportional to $\max\{d, \log_2\varepsilon^{-1}\}  (\log_2 d + \log_2 \varepsilon^{-1})^2$ and 
a number of quantum operations proportional to $\max\{d, \log_2\varepsilon^{-1}\} (\log_2 d + \log_2 \varepsilon^{-1})^3$.
It can be shown that
$\log_2 d = O(\log_2 \varepsilon^{-1})$ and the above expressions are simplified to
$\max\{d, \log_2\varepsilon^{-1}\} (\log_2 \varepsilon^{-1})^2$ qubits and $\max\{d, \log_2\varepsilon^{-1}\} ( \log_2 \varepsilon^{-1} )^3$ quantum operations. A measurement
outcome at the final state determines whether the algorithm has succeeded or not. A number
of repetitions proportional to the square of the condition number 
yields a success probability arbitrarily close to one.

In Section \ref{sec:DP} we deal with the discretization of the Poisson equation showing the resulting matrix. 
We also describe how the matrix in the 
multidimensional case can be expressed in terms of the one dimensional matrix using Kronecker products. This, as we'll see, is important in the simulation
of the Poisson matrix. In Section \ref{sec:QC} we show the quantum circuit solving the Poisson equation. 
We perform the error analysis and show the quantum circuit modules 
computing the reciprocal of the eigenvalues and from those the controlled rotation needed at the end of the linear 
systems algorithm \cite{lloyd09}.
In Section \ref{sec:HS} we deal with the Hamiltonian simulation of the matrix of the Poisson equation. 
The exponential of the multidimensional Hamiltonian is 
the $d$-fold tensor product of the exponential of one dimensional Hamiltonian. 
It is possible to diagonalize the one dimensional Hamiltonian using the quantum Fourier transform. Thus it
it suffices to approximate the eigenvalues in a way leading to the desired accuracy in the result. 
We show the quantum circuit modules performing the eigenvalue
approximation and derive the overall simulation cost. 
In Section \ref{sec:TotCost} we derive the total cost for solving the the Poisson equation.
Section \ref{sec:Concl} is the conclusion.
In Appendix 1 we list a number of elementary quantum gates and in Appendix 2 we present a series of results concerning the accuracy and the
cost of the approximations we use throughout the paper. 

\begin{figure}
\centerline{
\includegraphics[scale=1.2]{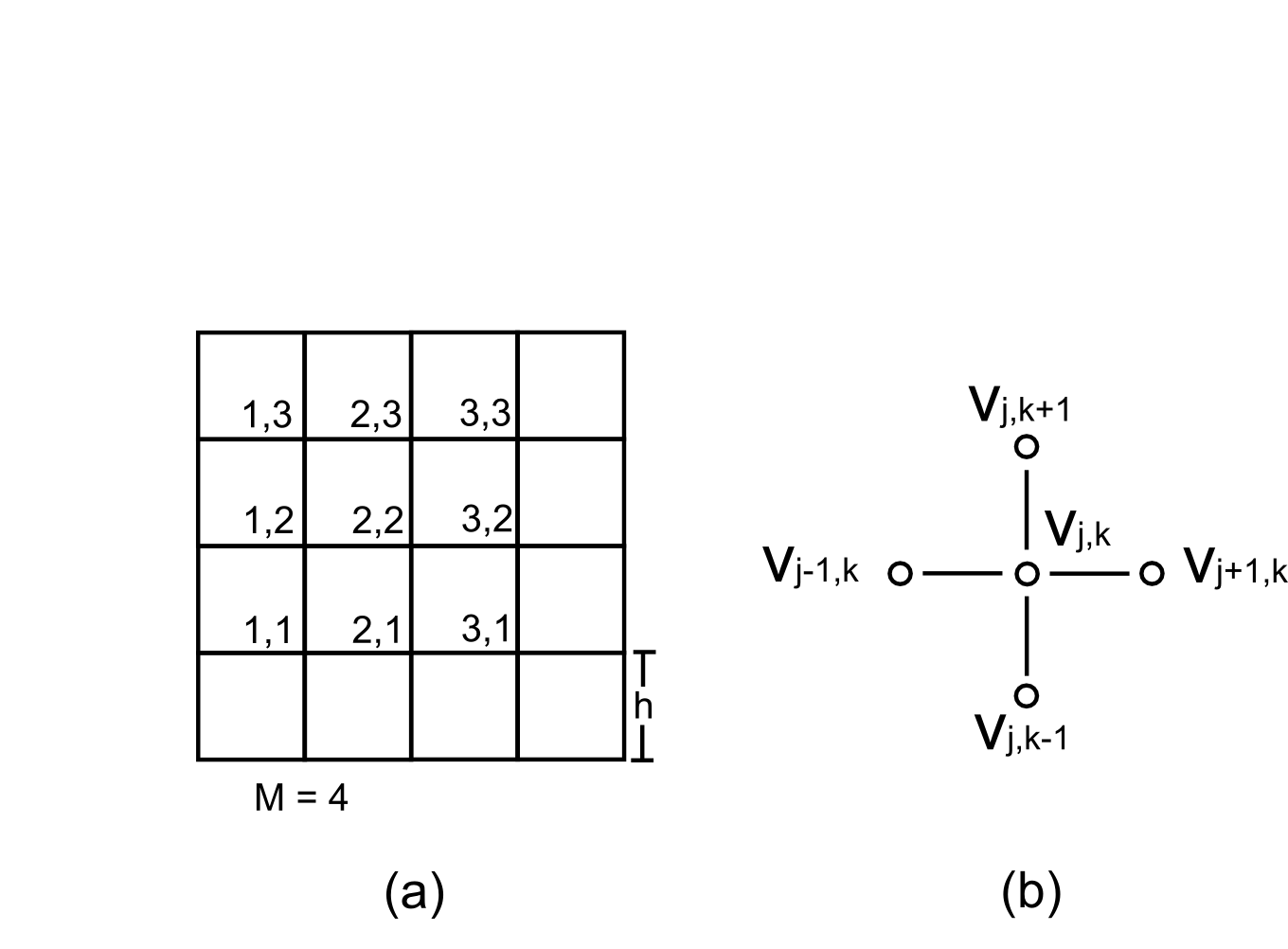} \\
}
\caption{Discretization of the square domain and notation for indexing the nodes.}
\label{fig:discrete}
\end{figure}

\section{Discretization} \label{sec:DP}

\subsection{One dimension}
We start with the one-dimensional case to introduce the matrix $L_h$ that we will use later in expressing the $d$-dimensional 
discretization of the Laplacian, using Kronecker products.
We have

\begin{eqnarray} 
&&-\frac{d^2u(x)}{dx^2}=f(x),  \text{$\quad$}x\in(0,1) \label{eq:1d}\\
&&u(0)=u(1)=0 \nonumber
\end{eqnarray}
\\*
where $f$ is a given smooth function and $u$ is the solution we want to compute. 
We discretize the problem with mesh size $h=1/M$ and we compute an approximate solution $v$ at $M+1$ grid points $x_i=ih$, $i=0,\dots,M$.
Let $u_i=u(x_i)$ and $f_i=f(x_i)$, $i=0,\dots,M$. 

Using finite differences at the grid points to approximate the second derivative (\ref{eq:1d}) becomes

\begin{equation} \label{eq:du2dx2}
-\frac{d^2u(x)}{dx^2}|_{x=x_i}=\frac{2u_i-u_{i-1}-u_{i+1}}{h^2}-\xi_i
\end{equation}
where $\xi_i$ is the truncation error and can be shown to be $O(h^2||\frac{d^4u}{dx^4}||_{\infty})$ if $f$ has fourth derivative uniformly bounded by a constant \cite{demmel97}. 

Ignoring the truncation error, we solve 
\begin{equation} \label{eq:discrete_1d}
h^{-2}(-v_{i-1}+2v_i-v_{i+1})=f_i \quad 0<i<M.
\end{equation}

With boundary condition $v_0=0$ and $v_M=0$, we have $M-1$ equations and $M-1$ unknowns $v_1,...,v_{M-1}$:

\begin{equation} \label{eq:matrix_1d}
h^{-2} \cdot L_h
\left( \begin{array}{c}
v_1 \\
\vdots \\
\vdots \\
v_{M-1} \end{array} \right)
:= h^{-2}  
\left( \begin{array}{cccc}
2 & -1 & \quad & 0 \\ 
-1 & \ddots & \ddots & \\ 
\quad & \ddots & \ddots & -1 \\ 
0 & \quad & -1 & 2 
\end{array} \right)
\left( \begin{array}{c}
v_1 \\
\vdots \\
\vdots \\
v_{M-1} \end{array} \right)
=
\left( \begin{array}{c}
f_1 \\
\vdots \\
\vdots \\
f_{M-1} \end{array} \right)
\end{equation}
where $L_h$ is the tridiagonal $(M-1)\times(M-1)$ matrix above; for the properties 
of this matrix, including its eigenvalues and eigenvectors see \cite[Sec. 6.3]{demmel97}. 

\subsection{Two dimensions}

In two dimensions the Poisson equation is

\begin{eqnarray} 
&&-\frac{\partial^2{u(x,y)}}{\partial{x^2}}-\frac{\partial^2{u(x,y)}}{\partial{y^2}}=f(x,y),\text{$\quad$}(x,y)\in (0,1)^2 \label{eq:2d}\\
&&u(x,0)=u(0,y)=u(x,1)=u(1,y)=0, \quad x,y\in [0,1] \nonumber
\end{eqnarray}

We discretize this equation using a grid with mesh size $h=1/M$; see
Figure \ref{fig:discrete}.
Each node is indexed $u_{j,k}$, $j,k\in\{1,2,...,M\}$ (Figure \ref{fig:discrete}(a) and (b)). 
We approximate the second derivatives using 
\begin{eqnarray*}
\frac{\partial^2u}{\partial x^2}(x,y)
&\approx&\frac{u(x-h,y)-2u(x,y)+u(x+h,y)}{h^2} \\
\frac{\partial^2u}{\partial y^2}(x,y)
&\approx&\frac{u(x,y-h)-2u(x,y)+u(x,y+h)}{h^2}.
\end{eqnarray*}

Omitting the truncation error, and denoting by $-\Delta_h$ the discretized Laplacian
we are led to solve 

\begin{equation} \label{eq:discrete}
h^{-2}\left((-v_{j-1,k}+2v_{j,k}-v_{j+1,k})+(-v_{j,k-1}+2v_{j,k}-v_{j,k+1})\right)=f_{j,k},
\end{equation}
where $f_{j,k} = f(jh,kh)$, $j,k = 1,2, \ldots, M-1$ and $v_{j,k} =0 $ if $j$ or $k \in \{0,M\}$ i.e., when we have 
a point that belongs to the boundary.

Using the fact that the solution is zero at the boundary, we reindex (\ref{eq:discrete}) to obtain
\begin{equation} \label{eq:discrete_transform}
h^{-2}(4v_i-v_{i-1}-v_{i+1}-v_{i-M+1}-v_{i+M-1})=f_i \quad i= 1,2,\ldots,(M-1)^2,
\end{equation}

Equivalently, we denote this system by 
\[-\Delta_h \vec v = \vec{f_h} ,\]
where $\Delta_h$ is the discretized Laplacian.

For example, when $M=4$, as in Figure \ref{fig:discrete}, 
we have that $\vec{v}=[v_1,...,v_9]^T$. Furthermore (\ref{eq:discrete_transform})  becomes

\begin{equation} \label{eq:2d_3by3}
h^{-2} A
\left( \begin{array}{c}
v_1 \\
\vdots \\
v_9 \end{array} \right)
:=h^{-2} 
\left( \begin{array}{ccc}
B & -I & \\
-I & B & -I \\
& -I & B 
\end{array} \right)
\left( \begin{array}{c}
v_1 \\
\vdots \\
v_9 \end{array} \right)=
\left( \begin{array}{c}
f_1 \\
\vdots \\
f_9 \end{array} \right),
\end{equation}
where $I$ is the $3 \times 3$ identity matrix, $B$ is 
\begin{equation*}
\left( \begin{array}{ccc}
4  & -1 & \\
-1 & 4  & -1 \\
   & -1 & 4  
\end{array} \right)
\end{equation*}
$A$ is a Hermitian matrix with a particular block structure that is independent of $M$. 

In particular, on a square grid with mesh size $h=1/M$ we have
\begin{equation}
\label{eq:Delta_h}
-\Delta_h = h^{-2} A  
\end{equation}
and $A$ can be 
expressed in terms of $L_h$ as follows:
\begin{equation} \label{eq:A_block}
A=
\left( \begin{array}{cccccc}
L_h+2I   & -I      & 0      & \cdots & \cdots & 0      \\
-I      & L_h+2I   & -I      & 0      & \cdots & 0      \\
0      & -I      & \ddots & \ddots & 0      & \vdots \\
\vdots & 0      & \ddots & \ddots & -I      & 0      \\
\vdots & \vdots & 0      & -I      & L_h+2I   & -I      \\
0      & 0      & \cdots & 0      & -I      & L_h+2I
\end{array} \right)
\end{equation}
and its size is $(M-1)^2\times(M-1)^2$ \cite{demmel97}.

Recall that $L_h$ is the $(M-1)\times(M-1)$ matrix shown in (\ref{eq:matrix_1d}) 
and $I$ is the $(M-1)\times(M-1)$ identity matrix. 
Moreover, $A$ can be expressed using Kronecker products as follows

\begin{equation} \label{eq:AT}
A=L_h\otimes{I}+I\otimes{L_h} .
\end{equation}

\subsection{$d$ dimensions}

We now consider the problem in $d$ dimensions. Consider the Laplacian 
\[\Delta = \sum_{k=1}^d \frac{\partial^2}{\partial x_k^2}.\]
We discretize $\Delta$ on a grid with mesh size $h=1/M$ using divided differences.

As before, this leads to a system of linear equations 
\begin{equation}
\label{eq:d-dim-general}
-\Delta_h\vec{v} = \vec{f_h}.
\end{equation}
Note that $-\Delta_h = h^{-2} A$ is symmetric positive definite matrix and $A$ is given by
\begin{equation*} 
A=\underbrace{L_h\otimes{I}\otimes \cdots \otimes{I}}_\text{$d$ matrices}+
I\otimes{L_h}\otimes I\otimes \cdots \otimes{I}+\cdots +
I\otimes\cdots \otimes I\otimes{L_h},
\end{equation*}
and has size $(M-1)^d\times(M-1)^d$. 
$L_h$ is the $(M-1)\times(M-1)$ matrix shown in (\ref{eq:matrix_1d}) and $I$ is the $(M-1)\times(M-1)$ identity matrix.
See \cite{demmel97} for the details.

Observe that the matrix exponential has the form
\begin{equation} \label{eq:A_d_dim}
e^{iA \gamma}=\underbrace{e^{iL_h\gamma}  \otimes \cdots \otimes e^{iL_h\gamma}}_{d \,\, {\rm matrices}},
\end{equation}
for all $\gamma \in \mathbb{R}$, where $i = \sqrt{-1}$.
We will use this fact later in deriving the quantum circuit solving the linear system.

\section{Quantum circuit} \label{sec:QC}
\label{sec:overview}

\begin{figure}
\centerline{
\Qcircuit @C=0.7em @R=0.5em @!R{
\lstick{Anc.} & & & |0\rangle\qquad & \qw & \qw & \qw & \qw & \qw & \qw & \qw & \qw & \qw & \qw & \qw & \qw & \qw & \qw & \qw & \gate{R_y} & \qw & & & & & & { \tilde{h}_j{\ket{1}}+\sqrt{1-\tilde{h}_j^2}\ket{0}} \\
\lstick{Reg.L} & & & |0\rangle\qquad & {/} \qw & \qw & \qw & \qw & \qw & \qw & \qw & \qw & \qw & \multigate{2}{\text{INV}} & \qw & & & |\hat{h}_j\rangle\quad & & \ctrl{-1} & \multigate{4}{U^\dagger} & \qw & & |0\rangle \\
{\quad^{b=3\lceil{\log\varepsilon^{-1}}\rceil\text{ qubits}}}\\
\lstick{Reg.C} & & & |0\rangle\qquad & {/} \qw & \gate{W} & \ctrl{2} & \gate{FT^\dagger} & \qw & & & |k_j\rangle\quad & & \ghost{\text{INV}} & {/} \qw & \qw & \qw & \qw & \qw & \qw & \ghost{U^\dagger} & \qw & & |0\rangle \\
{\quad^{n=O(\log(E/\varepsilon))\text{ qubits}}}\\
\lstick{Reg.B} & & & |f_h\rangle\qquad & {/} \qw & \qw & \gate{\text{HAM-SIM}} & \qw & & \sum_j\beta_j|u_j\rangle & & & & \qw & \qw & \qw & \qw & \qw & \qw & \qw & \ghost{U^\dagger} & \qw & & |b\rangle \\
}}
\caption{Overview of the circuit for solving the Poisson equation. Wires with `/' represent registers or groups of qubits. $W$ denotes the Walsh-Hadamard transform which applies Hadamard gate on every qubit of the register. $FT$ represents the quantum Fourier transform. \ \lq HAM-SIM\rq\  is the Hamiltonian Simulation subroutine that implements the operation $e^{-2\pi{i}\Delta_h/E}$. \lq INV\rq\  is the subroutine that computes $\lambda^{-1}$. $U^\dagger$ represent uncomputation, which is the adjoint of all the operations before the controlled $R_y$ rotation.}
\label{fig:general_circuit}
\end{figure}

We derive a quantum circuit solving the system $-\Delta_h \vec v = \vec{f_h}$, where $h = 1/M$ and
without loss of generality we assume that $M$ is a power of two. We obtain a solution of the system
with error $O(\varepsilon)$. The steps below are similar to those in~\cite{lloyd09}:

\begin{enumerate}
\item As in \cite{lloyd09} assume the right hand side vector $\vec{f_h}$ has been prepared quantum mechanically 
as a quantum state $\ket{f_h}$ and stored in the quantum register $B$.  Note $|f_h\rangle=\sum_{j=0}^{(M-1)^d-1}\beta_j|u_j\rangle$ where $|u_j\rangle$ denote the eigenstates of $-\Delta_h$ and $\beta_j$ are the coefficients.
\item Perform phase estimation using the state $\ket{f_h}$ in the bottom register and
the unitary matrix $e^{-2\pi i\Delta_h /E}$, where 
$\log_2 E = \lceil \log d \rceil + \log (4M^2)$. 
The number of qubits in the top register of phase estimation is $n = O(\log (E/\varepsilon))$.
\item Compute an approximation of the inverse of the eigenvalues $\lambda_j$. Store
the result on a register $L$ composed of $b = 3\lceil \log\varepsilon^{-1}\rceil$ qubits (Figure~\ref{fig:general_circuit}). 
The approximation error of the reciprocals is at most $\varepsilon$. 
\item Introduce an ancilla qubit to the system. Apply a controlled rotation
on the ancilla qubit. The rotation operation is controlled be the register $L$
which stores the reciprocals of the eigenvalues of $-\Delta_h$ (Figure~\ref{fig:general_circuit}). The controlled rotation
results to $\sqrt{1 - (C_d/\lambda_j)^2} \ket{0} + (C_d/\lambda_j) \ket{1}$, where $C_d$ is a constant.
\item Uncompute all other qubits on the system except the qubit introduced on the 
previous item.
\item Measure the ancilla qubit. If the outcome is $1$, the bottom register of 
phase estimation collapses to the state $\sum_{j=0}^{(M-1)^d-1} \beta_j{\lambda_j}^{-1} \ket{u_j}$ up to a normalization factor,
where $\ket{u_j}$ denote the eigenstates of $-\Delta_h$. This is equal to the normalized solution of the system.
If the outcome is $0$, the algorithm has failed and we have to repeat it. An alternative would be to include
amplitude amplification to boost the success probability. Amplitude amplification has been considered 
in the literature extensively and we do not deal with it here.
\end{enumerate}

\subsection{Error analysis}
\label{sec:Analysis}

We carry out the error analysis to obtain the implementation details.
For $d=1$ the eigenvalues of the second derivative are
\[4 M^2 \sin^2(j\pi/(2M)) \quad j=1,\dots,M-1.\]
For $d>1$, the eigenvalues of $-\Delta_h$ are given by sums of the one-dimensional eigenvalues, i.e.,
\[\sum_{k=1}^d \left[4 M^2 \sin^2(j_k\pi/(2M))\right] \quad j_k=1,\dots,M-1, \; k=1,\dots,d.\]
We consider them in non-decreasing order and denote them by $\lambda_j$, $j=1,\dots, (M-1)^d.$
Then $\lambda_1 = 4dM^2 \sin^2(\pi/(2M))$ is the minimum eigenvalue and 
$\lambda_{(M-1)^d} = 4dM^2 \sin^2(\pi (M-1)/(2M))\le 4 d M^2$ is the maximum eigenvalue.

Define $E$ by 
\begin{equation}
\label{eq:def-E}
\log_2 E=\lceil \log_2 d\rceil + \log_2 (4M^2).
\end{equation}

Then the eigenvalues are bounded from above by $E$. Recall that we have already assumed that $M$ is a power of two. Then $E=2^{\lceil{\log_2{d}}\rceil}4M^2\in\mathbb{N}$.

Note that the implementation accuracy of the eigenvalues determines the accuracy of the system solution.

Our algorithm uses approximations $\hat \lambda_j$, such that $|\lambda_j - \hat\lambda_j|\leq \frac{17 \cdot E}{2^\nu} \leq \varepsilon$; see Theorem \ref{th:d-error-eig} in Appendix 2. We use $n = \log_2 E + \nu$ bits to
represent each eigenvalue, of which the $\log_2 E$ most significant bits hold each
integer part and the remaining bits hold each fractional part. Without
loss of generality, we can assume that $2^\nu \gg E$.
More precisely, we consider an approximation $\hat \Delta_h$ of matrix $\Delta_h$ such that
the two matrices have the same eigenvectors while their eigenvalues differ by at most $\varepsilon$.

We use phase estimation with the unitary matrix $e^{-i\hat \Delta_h t_0 /E}$ 
whose eigenvalues are $e^{2\pi i \hat \lambda_j t_0/(E 2\pi)}$. Setting $t_0=2\pi$ we
obtain the phases $\phi_j = \hat \lambda_j /E \in [0,1)$.
The initial state of phase estimation is (Figure~\ref{fig:general_circuit})
\[\ket{0}^{\otimes n}\ket{f_h} = \sum_{j=1}^{(M-1)^d} \beta_j \ket{0}^{\otimes n}\ket{u_j},\]
where $\ket{u_j}$ is the $j$th eigenvector of $-\Delta_h$ and $\beta_j = \Braket{u_j}{f_h}$,
for $j=1,2,\ldots ,(M-1)^d$.
Since we are using finite bit approximations of the eigenvalues, we have
\[\phi_j = \frac{\hat \lambda_j}{E}= 
\frac{\hat \lambda_j 2^\nu}{2^n}.
\]
Then $\phi_j \, 2^n$ is an integer and phase estimation succeeds with probability $1$ (see~\cite[Sec. 5.2, pg. 221]{nielsen00} for details).

The state prior to the application of the inverse Fourier transform in phase estimation is
\begin{equation}
\label{eq:prior-F}
\sum_{j=1}^{(M-1)^d} \beta_j \frac{1}{2^{n/2}} \sum_{k=0}^{2^n-1} e^{2\pi i \phi_j k} \ket k \ket{u_j}.
\end{equation}
After the application of the inverse Fourier transform to the first $n$ qubits we obtain
\[\sum_{j=1}^{(M-1)^d} \beta_j \ket{k_j}\ket{u_j},\]
where 
\begin{equation}
\label{eq:kj}
k_j = 2^n \phi_j = 2^n  \hat\lambda_j / E=\hat\lambda_j2^\nu\in\mathbb{N}
\end{equation}

Now we need to compute the reciprocals of the eigenvalues. Observe that  
\begin{eqnarray*}
\lambda_1/d &=& 4M^2 \sin^2(\pi/(2M)) = 4 M^2 ( \pi / (2M) + O(M^{-3}))^2 \\
&=& \pi^2 + O(M^{-2}) > 5.
\end{eqnarray*}
where the last inequality holds trivially for $M$ sufficiently large.
This implies $\hat \lambda_j / C_d \geq \hat \lambda_1 / C_d \geq 4$, where 
$C_d = 2^{\lfloor \log_2 d\rfloor}$, for 
$M$ sufficiently large.
We obtain $k_j = 2^n \hat \lambda_j / E \geq \hat \lambda_1 \geq 4 C_d$.

Append $b$ qubits initialized to $\ket 0$ on the left ($Reg.L$ in Figure \ref{fig:general_circuit}), to obtain 
\[\sum_{j=1}^{(M-1)^d} \beta_j \ket{0}^{\otimes b}\ket{k_j}\ket{u_j}.\]
Note that from (\ref{eq:kj}) $k_j$, $\hat \la_j$ and $\hat \lambda_j / C_d$ have the 
same bit representation. The difference between the integer $k_j$ and the other two numbers
is the location of the decimal point; it is located after the $\log_2 E$ most significant
bit in $\hat \la_j$, and after the $\log_2 (E/C_d)$ most significant bit in 
$\hat \lambda_j/C_d$. Therefore, we can use the labels $\ket{k_j}$, $\ket{\hat \lambda_j}$
and $\ket{\hat{\lambda}_j / C_d}$ interchangeably, and write the state above as
\[\sum_{j=1}^{(M-1)^d} \beta_j \ket{0}^{\otimes b}\ket{\hat \lambda_j/C_d}\ket{u_j}.\]
Now we need to compute $h_j := h (\hat \lambda_j / C_d) = C_d / \hat\lambda_j$. We do this using Newton 
iteration. We explain the details in Section \ref{sec:Newton}.
We obtain an approximation 
$\hat h_j$
such that 
\begin{equation}
\label{eq:error_Newton}
\left| \hat h_j - h_j \right| \leq \varepsilon_0^2,
\end{equation}
where $\varepsilon_0 = \min \{\varepsilon, E^{-1}\}$. We store this approximation in the register
composed of the leftmost $b = 3 \lceil \log_2 \varepsilon_0^{-1} \rceil$ qubits.

This leads to the state 
\[\sum_{j=1}^{(M-1)^d} \beta_j \ket{\hat h_j}\ket{\hat \la_j/C_d}\ket{u_j}.\]
We append, on the left, a qubit initialized at $\ket 0$ ($Anc.$ in Figure~\ref{fig:general_circuit}). We get
\[\sum_{j=1}^{(M-1)^d} \beta_j \ket 0 \ket{\hat h_j}\ket{\hat \la_j/C_d}\ket{u_j}.\]
We need to perform the conditional rotation 
\[R \ket 0 \ket{\omega}  = \left( \omega \ket 1 + \sqrt{1 - \omega^2} \ket 0 \right) \ket{\omega}, \quad
0<\omega < 1.\]
For this, we will approximate the first qubit by
\[\omega^\prime \ket 1 + \sqrt{1 - (\omega^\prime)^2} \ket 0,\]
with $|\omega -\omega^\prime|\le \varepsilon_1^2$, $\varepsilon_1= \min\{ \varepsilon, 1/(4M^2)\}$.
We discuss the cost of implementing this approximation in Section \ref{sec:R}.

The result of approximating the conditional rotation is to obtain
$\ket{\tilde h_j}$, where $\tilde h_j$ is 
a $q=\Theta(\log_2 \varepsilon^{-1}_1)$ bit number less than $1$ satisfying $|\tilde h_j - \hat h_j|\le\varepsilon_1^2$ and, therefore,
\begin{equation}
\label{eq:NewPlusCondRot}
|\tilde h_j - h_j|\le \varepsilon_0^2 + \varepsilon_1^2,
\end{equation}
for each $j=1,\dots,(M-1)^d$.

Ignoring the ancilla qubits needed for implementing the approximation of the conditional rotation, we have the state 
\[\sum_{j=1}^{(M-1)^d} \beta_j \left( \tilde h_j \ket 1 + \sqrt{1 - \tilde h_j^2} \ket 0 \right)
 \ket{\hat h_j}\ket{\hat \la_j/C_d}\ket{u_j}.\]
 Uncomputing all the qubits except the leftmost gives the state
\[\ket \psi : = \sum_{j=1}^{(M-1)^d} \beta_j  \left( \tilde h_j \ket 1 + \sqrt{1 - \tilde h_j^2} \ket 0 \right)
  \ket{0}^{\otimes b} \ket{0}^{\otimes n} \ket{u_j}\]  

Let $P_1=\Ketbra{1}{1}\otimes I$ be the projection acting non-trivially on the first qubit. 
The system $-\Delta_h \vec v = \vec f_h$ has solution $\sum_{j=1}^{(M-1)^d} \beta_j\frac{1}{\la_j} \ket{u_j}$. 
We derive the error as follows
\begin{eqnarray}
\label{eq:syserr}
&& C_d^{-1}\left\| \sum_{j=1}^{(M-1)^d} b_j \frac{C_d}{\lambda_j} \ket{1} \ket{0}^{\otimes{(b+n)}} \ket{u_j} -
P_1\ket{\psi} \right\| =  \nonumber\\ 
&&
C_d^{-1}\left\|\sum_{j=1}^{(M-1)^d} \beta_j \frac{C_d}{\lambda_j} \ket{u_j} - \sum_{j=1}^{(M-1)^d} \beta_j \tilde h_j \ket{u_j}\right\| 
 = \nonumber \\
&&
C_d^{-1}\left\|\sum_{j=1}^{(M-1)^d} \beta_j \frac{C_d}{\lambda_j} \ket{u_j} - \sum_{j=1}^{(M-1)^d} \beta_j (\tilde h_j - 
h_j + h_j)\ket{u_j}\right\| 
 \leq  \nonumber \\
&&
\quad\quad \left\|\sum_{j=1}^{(M-1)^d} \beta_j \left( \frac{1}{\lambda_j} - \frac{1}{\hat \lambda_j} \right)\ket{u_j} \right\| 
+ \varepsilon_0^2 + \varepsilon_1^2 \leq \frac{17 E}{2^\nu} + \varepsilon_0^2 +\varepsilon_1^2,
\end{eqnarray}
where the second from last inequality is obtained using (\ref{eq:NewPlusCondRot}) and the last inequality is due to the 
fact that 
\[ \left|\frac{1}{\lambda} - \frac{1}{\hat\lambda}\right| \leq |\lambda - \hat \lambda|, \quad
\lambda , \hat \lambda > 1\]
Setting $\nu = \lceil \log_2 (17 E /\varepsilon)\rceil$ gives error $\varepsilon (1+o(1))$
and the number of matrix exponentials used by the algorithm is $O(\log_2 (E / \varepsilon))$.
Therefore, if we measure the first qubit of the state $\ket \psi$ and the outcome is $1$ the state collapses to 
a normalized solution of the linear system. 

\subsection{Computation of $\la^{-1}$}
\label{sec:Newton}

In this part we deal with the computation of the reciprocals of the 
eigenvalues, which is marked as the \lq INV\rq\  module in Figure~\ref{fig:general_circuit}. For this we use Newton iteration to approximate $v^{-1}$, $v>1$.
We perform $s$ iterative steps and 
obtain the approximation $\hat x_s$. The input and the output of
each iterative step are $b$ bit numbers. All the calculations in each step
are performed in fixed precision arithmetic.
The initial approximation is $\hat x_0 = 2^{-p}$, $2^{p-1} < v \le 2^p$.
(We use the notation $\hat x_i$ to emphasize that these values have been
obtained by truncating a quantity $x_i$ to $b$ bits of accuracy, $i = 0,\ldots ,s$).

Theorem \ref{th:Newton_error} of Appendix 2 gives the error of Newton iteration which is
\[|\hat x_s - v^{-1} | \leq \varepsilon_0^2 \leq \varepsilon,\]
where we have $\varepsilon_0=\min\{ \varepsilon, E^{-1}\}$,  
$s= \lceil \log_2\log_2 (2/\varepsilon_0^2)\rceil$ and the number of bits 
satisfies
$b\geq 2\lceil \log_2 \varepsilon_0^{-1}\rceil + O(\log_2 \log_2 \log_2 \varepsilon_0^{-1})$.

Therefore, it suffices that the module of the quantum circuit that computes $1/\lambda_j$ carries each iterative step 
with $3\lceil \log_2 \varepsilon_0^{-1}\rceil$ qubits of accuracy. 

The quantum circuit computing the initial approximation $\hat x_0$, 
of the Newton iteration is given in Figure~\ref{fig:x0_circuit}.
The second register holds $\ket v$ and is $n$ qubits long,
of which the first $\log_2 (E /C_d)$ qubits 
represent the integer part of $v$ and the remaining ones its fractional part. 
The first register is $b$ qubits long. 
Recall that $\hat \la_j / C_d \geq 4$. So input values below $4$ do not correspond to meaningful
eigenvalue estimates and we don't need to compute their reciprocals altogether; they can be ignored.
Hence the circuit implements the unitary transformation $\ket{0}^{\otimes b} \ket v \to  \ket{0}^{\otimes b} \ket v$, 
if the first 
$\log_2 (E /C_d) -2$ bits of $v$ are all zero. 
Otherwise, it implements the initial approximation $\hat x_0$  through the transformation
$\ket{0}^{\otimes b} \ket v \to \ket{\hat x_0} \ket v$. 

\begin{figure}
\Qcircuit @C = 2.3em @R = 2em {
  & & & \lstick{\ket{u_{b-1}=0}}    & \qw  & \qw    & \qw   & \dots & & \targ    & \ctrl{6} & \qw       & \qw & \qw \gategroup{1}{14}{6}{14}{1.2em}{\}} & & \\
  & & & \vdots &      \\
  \raisebox{-3em}{$b$}& & & \lstick{\ket{u_{b^\prime + 1}=0}}    & \qw   & \targ & \ctrl{4} & \dots & & \qw   & \qw  & \qw & \qw & \qw &\raisebox{-3em}{$\ket{x_0}$}  &  \\
  & & & \lstick{\ket{u_{b^\prime}=0}}& \targ & \qw & \qw      & \dots & & \qw    & \qw   & \qw   & \qw & \qw \gategroup{1}{2}{6}{2}{1.2em}{\{} & & \\
  & & & _{b^\prime = b - 2 - \log_2 \frac{E}{C_d}} & & & & \vdots &   \\
  & & & \lstick{\ket{u_{0}=0}}  & \qw  & \qw &\qw   & \dots & & \qw & \qw  & \qw  & \qw & \qw \gategroup{8}{2}{13}{2}{1em}{\{} & & \\
& & & \lstick{\ket{0}}     & \targ &\ctrlo{-4} & \targ   & \cdots & &\ctrlo{-6} &\targ   &\targ & \gate{X} & \qw & \ket{0} & \\
  & & & \lstick{\ket{v_{n-1}}}& \ctrl{-4} & \qw & \qw   &\dots &  & \qw &\qw  &\ctrlo{-1} & \qw & \qw & & \\ 
  & & & \lstick{\ket{v_{n-2}}}& \qw       & \ctrl{-2}  & \qw   &\dots &  & \qw &\qw  &\ctrlo{-1} & \qw & \qw & \raisebox{-6em}{$\ket v$} & \\
  \raisebox{-2em}{$n$} & & & & & & & \vdots & \\
  & & & \lstick{\ket{v_{n^\prime}}}     & \qw       & \qw & \qw   & \dots &  & \ctrl{-4}  & \qw     & \ctrlo{-2}  & \qw & \qw & & \\
  & & & _{n^\prime = n+2-\log_2 \frac{E}{C_d}} & & & & \vdots &\\
  & & & \lstick{\ket{v_{0}}}  & \qw  &  \qw & \qw   & \dots & & \qw  & \qw   &\qw        & \qw & \qw \gategroup{8}{14}{13}{14}{0.7em}{\}}& &
}
\caption{The quantum circuit computing the initial approximation $\hat{x}_0=2^{-p}$ of Newton iteration
for approximating $v^{-1}$, $2^{p-1}\le v\le 2^p$. See Appendix 1 for definitions of the basic gates.}
\label{fig:x0_circuit}
\end{figure}
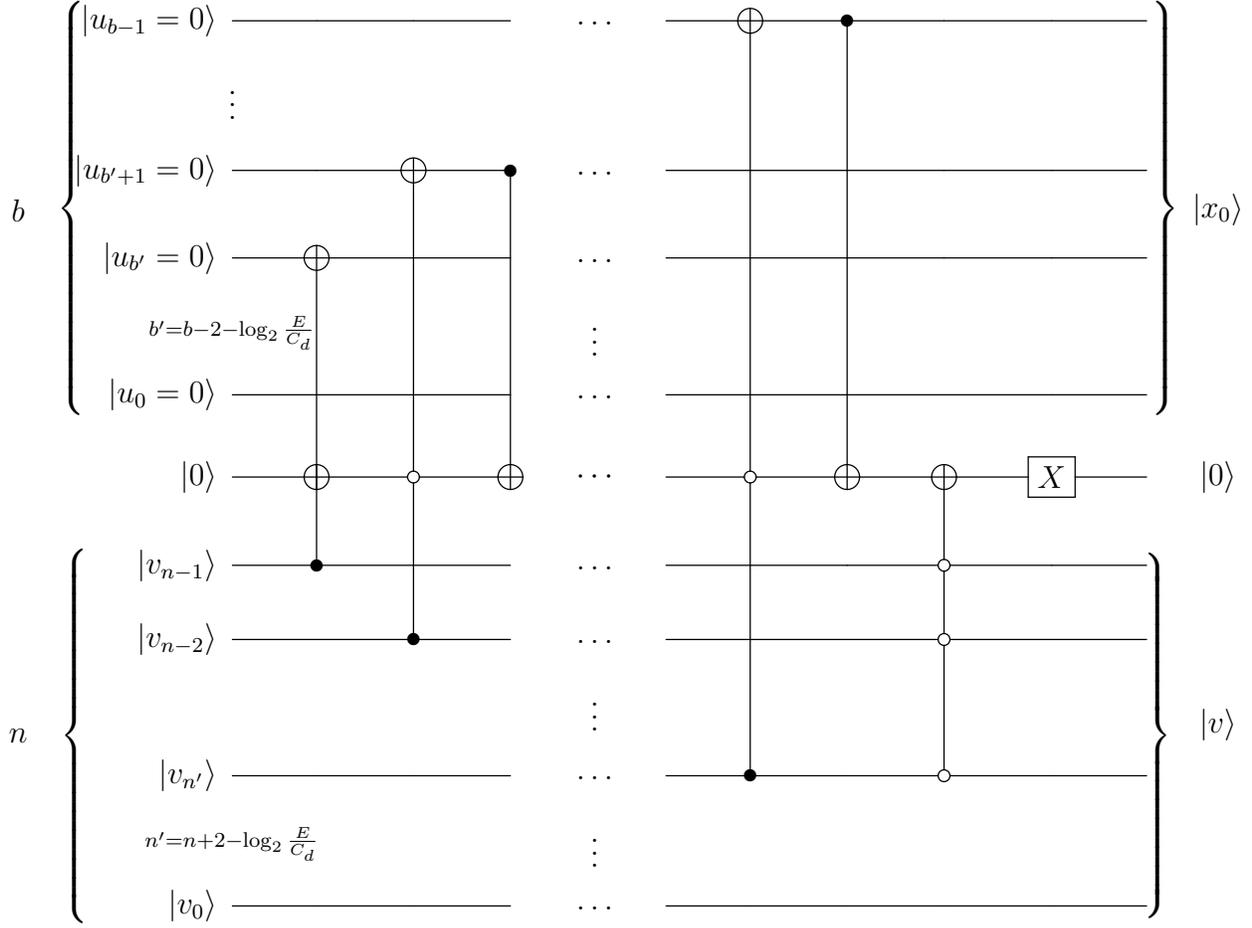

Each iteration step $x_{i+1} = -vx_{i}^2 + 2 x_i$ 
is implemented using a quantum circuit of the form shown in Figure~\ref{fig:iteration} that computes 
$\ket{\hat x_{i}} \ket v \to \ket{\hat x_{i+1}} \ket v$. 
This involves quantum circuits for addition and multiplication which have been studied in the literature \cite{vedral96}.

\begin{figure}
\[\Qcircuit @C = 2.7em @R = 2em {
\lstick{\ket{\hat x_i}}  & \multigate{1}{-vx_i^2 + 2x_i} & \qw &  \ket{\hat x_{i+1}} \\
\lstick{\ket{v}} & \ghost{-vx_i^2 + 2x_s}        & \qw &  \ket{v} 
}\]
\caption{Circuit implementing each iterative step of the Newton method.}
\label{fig:iteration}
\end{figure}
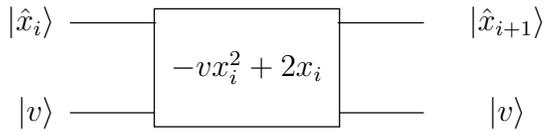

The register holding $\ket{v}$ is $n$ qubits long and the register holding the $\ket{\hat x_i}$ and $\ket{\hat x_{i+1}}$ is
$b$ qubits long.
Note that internally the modules performing the iteration steps may use more than $b$ qubits, say,
double precision, so that
the addition and multiplication operations required in the iteration are carried out exactly and then 
return the $b$ most significant 
qubits of the result. The total number of qubits required for the implementation of each of
these modules is $O(\log \varepsilon_0^{-1})$ and the
total number of gates is a low degree polynomial in $\log \varepsilon_0^{-1}$.

\subsection{Controlled rotation}
\label{sec:R}

We now consider the implementation of the controlled rotation
\[R \ket 0 \ket{\omega}  = \left( \omega \ket 1 + \sqrt{1 - \omega^2} \ket 0 \right) \ket{\omega}, \quad
0<\omega < 1.\]
Assume for a moment that we have obtained $\ket{\theta}$, a $q$ qubit state, corresponding to an angle $\theta$ 
such that $\sin \theta$ approximates $\omega$. Then we can use
controlled rotations $R_y$ about the $y$ axis to implement $R$. 
We consider the  binary representation of $\theta$ and have 
\[\theta = .\theta_1 \ldots \theta_q = \sum_{j=1}^q \theta_j 2^{-j}, \quad \theta_j \in \{0,1\}.\]
Then
\begin{eqnarray*}
 R_y(2\theta)  &=& e^{-i \theta Y} = \left(  \begin{array}{cc}
\sqrt{1-\sin^2\theta} & -\sin\theta \\
\sin \theta & \sqrt{1-\sin^2\theta}
\end{array} \right) \\
&=&  \prod_{j=1}^q e^{-i Y \theta_j / 2^j} 
= \prod_{j=1}^q R_y^{\theta_j} \left( 2^{1-j}\right) ,
\end{eqnarray*}
where $Y$ is the Pauli $Y$ operator and $\theta\in[0,\pi/2]$. The detailed circuit is shown in Figure~\ref{fig:ctrl_Ry}.

\begin{figure}
\[\Qcircuit @R =2.3em @C=2em{
\ket 0 &  & \gate{R_y(1)} & \gate{R_y(1/2)} & \qw & \dots & & \gate{R_y(1/2^{q-1})} & \qw \\
\ket{\theta_1} &  & \ctrl{-1} & \qw & \qw & \dots & & \qw & \qw \\
\ket{\theta_2} &  & \qw & \ctrl{-2} & \qw & \dots & & \qw & \qw \\
\vdots & \\
\ket{\theta_b} &  & \qw & \qw & \qw & \dots & & \ctrl{-4} & \qw 
}\]
\caption{Circuit for executing the controlled $R_y$ rotation. See Appendix 1 for definitions of basic gates.}
\label{fig:ctrl_Ry}
\end{figure}
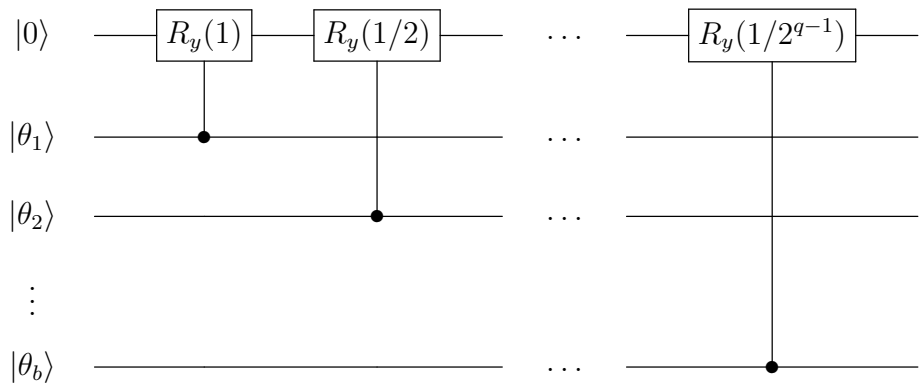

We now turn to the algorithm that calculates $\ket{\theta}$ from $\ket{\omega}$.
Since $\omega$ corresponds to the reciprocal of an approximate eigenvalue of the discretized Laplacian,
we know that $\sin^{-1}(\omega)$ belongs to the first quadrant and  $\sin^{-1}(\omega)=\Omega(1/M^2)$. 
Therefore, we can find an angle $\theta$ such that $|\sin(\theta) - \omega|\le \varepsilon_1^2$, $\varepsilon_1=\min\{ 1/(4M^2),\varepsilon \}$, using 
bisection and an approximation of the sine function. 

In Appendix 2 we show the error in approximating the sine function
using fixed precision arithmetic. 
In Section \ref{sec:HS} we show the details of the resulting quantum algorithm computing the approximation to the sine function.
These results, with a minor adjustment in the number of bits needed can be used here. 
We won't deal with the details of the quantum algorithm for the sine function
in this section since we present them in Section \ref{sec:HS} that deals with the simulation of Poisson's matrix. 
We will only describe the steps of the algorithm and its cost.

\vskip 1pc
\noindent Algorithm:
\begin{enumerate}
\item Take as an initial approximation of $\theta$ the value $\pi/4$.
\item Approximate the $\sin(\theta)$ with error $\varepsilon_1^2/2$ using our algorithm for the sine function (details in section~\ref{sec:HS} and Appendix 2). Let $s_\theta$ denote this approximation.
\item If $s_\theta < \omega - \varepsilon_1^2/2$, set $\theta$ to be the midpoint of the right subinterval.
\item If $s_\theta > \omega + \varepsilon_1^2/2$, set $\theta$ to be the midpoint of the left subinterval.
\item Repeat the steps 2 to 4 $\lceil \log_2\varepsilon_1^{-2}\rceil + 1$ times.
\end{enumerate}

An evaluation at the midpoint of an interval yields a value that satisfies either the condition of step 3, or that of step 4,
or $|s_\theta-\omega|\le \varepsilon_1^2/2$.
If at any time both the conditions of steps 3 and 4 are false then $\theta$ will not change its value until the end.
Then, at the end, 
we have $|\sin(\theta)-\omega|\le |\sin(\theta)- s_\theta| + |s_\theta - \omega|\le \varepsilon_1^2$, since the error in computing the sine is $\varepsilon_1^2/2$.
On the other hand, 
if $\theta$ is updated until the very end of the algorithm the final value of theta also satisfies
$|\sin(\theta)-w|\le \varepsilon_1^2$, because in the final interval we have 
$|\sin(\theta)-\omega| \le |\theta - \sin^{-1}(\omega)|\le \varepsilon_1^2$.

In a way similar to that of Proposition 1 and Proposition 2 of Appendix 2 we carry out the steps of the algorithm in
$q$ bit fixed precision arithmetic,
$q=\max\{ 2\nu + 9, 13 + \nu + 2\log_2 M\}$ and sufficiently large $\nu$ to satisfy the accuracy requirements.
(The last expression for $q$  is slightly different form that in Proposition 2 because it accounts for the
fact that in the case we are dealing with here the angle is $\Omega(1/M^2)$). This gives us an approximation to the sine with error $2^{-(\nu-1)}$.
We set
\[\nu = \lceil \log_2 \varepsilon_1^{-2} \rceil+ 1.\]
Thus $\nu$ and $q$ are both $\Theta(\log_2\varepsilon_1)$.

The algorithm for the sine function is based on an approximation of the exponential function using repeated squaring.
Each square requires $O(q^2)$ quantum operations and $O(q)$ qubits. This is repeated $\nu$ times before
the approximation to the sine is obtained. Thus the cost of one bisection step requires $O(\nu q^2)$ quantum operations and
$O(\nu q)$ qubits. So, in terms of $\varepsilon_1$, the total cost of bisection is proportional to
$(\log_2 \varepsilon_1^{-1})^4$ quantum operations and $(\log_2 \varepsilon_1^{-1})^3$ qubits.

\section{Hamiltonian simulation of the Poisson matrix} \label{sec:HS}

\begin{figure} 
\centerline{
\Qcircuit @C=2.3em @R=2em @!R{
 {|1\rangle} & & \multigate{1}{e^{ih^{-2} L_h\gamma}} & \qw & \qw \\
 \lstick{\ket{j_1}} & \qw {/} & \ghost{e^{ih^{-2} L_h\gamma}} & \qw {/} & \qw \\
 {|1\rangle} & & \multigate{1}{e^{ih^{-2} L_h\gamma}} & \qw & \qw \\
 \lstick{\ket{j_2}} & \qw {/} & \ghost{e^{ih^{-2} L_h\gamma}} & \qw {/} & \qw \\
}
}
\caption{Quantum circuit for implementing $e^{-i\Delta_h \gamma}$, $\gamma \in \mathbb{R}$ for the two dimensional discrete Poisson equation. The subroutine of $e^{ih^{-2} L_h  \gamma}$ is shown in Figure \ref{fig:eAt_1d}. The registers holding $\ket{j_1}$, $\ket{j_2}$ are $m$ qubits each.}
\label{fig:eAt}
\end{figure}
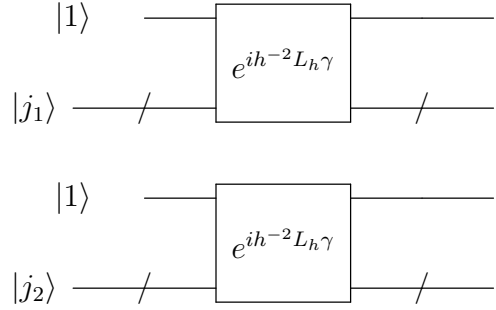

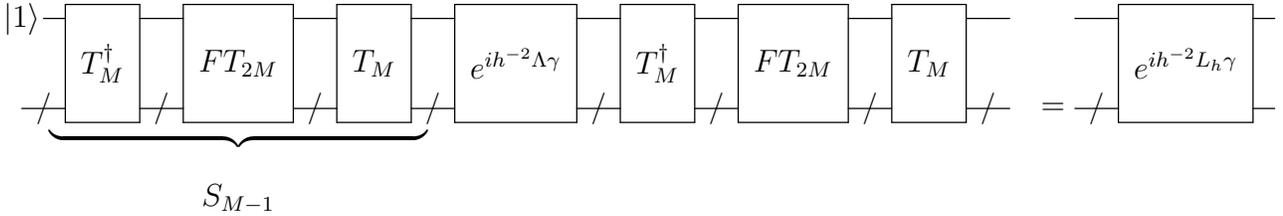
\begin{figure} 
\centerline{
\Qcircuit @C=0.7em @R=2em @!R{
 {|1\rangle} & & \multigate{1}{T_M^\dagger} & \qw & \multigate{1}{FT_{2M}} & \qw & \multigate{1}{T_M} & \qw & \multigate{1}{e^{ih^{-2}\Lambda{\gamma}}} & \qw & \multigate{1}{T_M^\dagger} & \qw & \multigate{1}{FT_{2M}} & \qw & \multigate{1}{T_M} & \qw & \qw & & & & \qw & \multigate{1}{e^{ih^{-2} L_h\gamma}} & \qw \\
 & \qw {/} & \ghost{T_M^\dagger} & \qw {/} & \ghost{FT_{2M}} & \qw {/} & \ghost{T_M} & \qw {/} & \ghost{e^{ih^{-2}\Lambda{\gamma}}} & \qw {/} & \ghost{T_M^\dagger} & \qw {/} & \ghost{FT_{2M}} & \qw {/} & \ghost{T_M} & \qw {/} & \qw  \gategroup{2}{3}{2}{7}{1em}{_\}} & & = & & \qw {/} & \ghost{e^{ih^{-2}L_h \gamma}} & \qw \\
& & & & \text{$S_{M-1}$} & & &
}
}
\caption{Quantum circuit for implementing $e^{ih^{-2}L_h \gamma}$, $\gamma \in \mathbb{R}$, where $L_h$ is the matrix in (\ref{eq:matrix_1d}). $S_{M-1}$ represents the sine transform matrix of size $(M-1)\times (M-1)$, $M=2^m$. This circuit acts on $m+1$ qubits.}
\label{fig:eAt_1d}
\end{figure}

\begin{figure}
\centerline{
\Qcircuit @C=1em @R=0.7em {
\lstick{|1\rangle } \gategroup{6}{3}{6}{4}{2em}{_\}}\gategroup{6}{5}{6}{6}{1.1em}{_\}} & \qw & \gate{B} & \gate{B^\dagger} & \ctrl{1} & \ctrl{1} & \qw \\
& \qw & \qw & \ctrlo{-1} \ctrlo{1} & \gate{X} & \multigate{4}{P_{m}} & \qw \\
& & & & & & & & & & & \lstick{|1\rangle} & \qw & \multigate{1}{T_M} & \qw & \qw \\
& & & \vdots & \vdots & & & & = & & & & \qw {/} & \ghost{T_M} & \qw {/} & \qw \\
& & & & & & \\
& \qw & \qw & \ctrlo{-1} & \gate{X} \qw[1] & \ghost{P_{m}} & \qw & \qw \\
\\
& & \text{\small ${\qquad\quad}D$} & & \text{\small $\qquad\qquad\pi$}
}
}
$\quad$ \\
$\quad$ \\
{(a) Generic circuit for $T_M=D\pi$, for details refer to \cite{roetteler01}. }
$\quad$ \\
\centerline{
\Qcircuit @C=1em @R=2em {
& \gate{B} & \qw & = & & \gate{H} & \gate{S} & \qw
}}
$\quad$ \\
\centerline{
\Qcircuit @C=1em @R=2em {
& \gate{X} & \qw & \qw & \qw & \cdots & & \qw & \qw & \qw & & & &  & \multigate{6}{P_m} & \qw \\
& \ctrl{-1} & \gate{X} & \qw & \qw & \cdots & & \qw & \qw & \qw & & & &  & \ghost{P_m} & \qw \\
& \ctrl{-2} & \ctrl{-1} & \gate{X} & \qw & \cdots & & \qw & \qw & \qw & & & &  & \ghost{P_m} & \qw \\
& \cdots & \cdots & \cdots & & \ddots & & & & & & = & & & \\
& & & & & & & & &  \\
& \ctrl{-1} & \ctrl{-1} & \ctrl{-1} & \qw & \cdots & & \gate{X} & \qw & \qw & & & &  & \ghost{P_m} & \qw \\
& \ctrl{-1} & \ctrl{-1} & \ctrl{-1} & \qw & \cdots & & \ctrl{-1} & \gate{X} & \qw & & & &  & \ghost{P_m} & \qw
}
}
$\quad$ \\
{(b) Implementation of $B$ and $P_m$ gates in (a).}
\caption{Quantum circuit for implementing $T_M$ in Equation \ref{eq:FSC} and \ref{eq:TN}. In (b), $P_m$ denotes the map $|x\rangle\rightarrow|x+1$ mod $2^n\rangle$ on $n$ qubits. Its implementation is described in \cite{pueschel98}. See Appendix 1 for the definitions of basic gates.}
\label{fig:TM}
\end{figure}
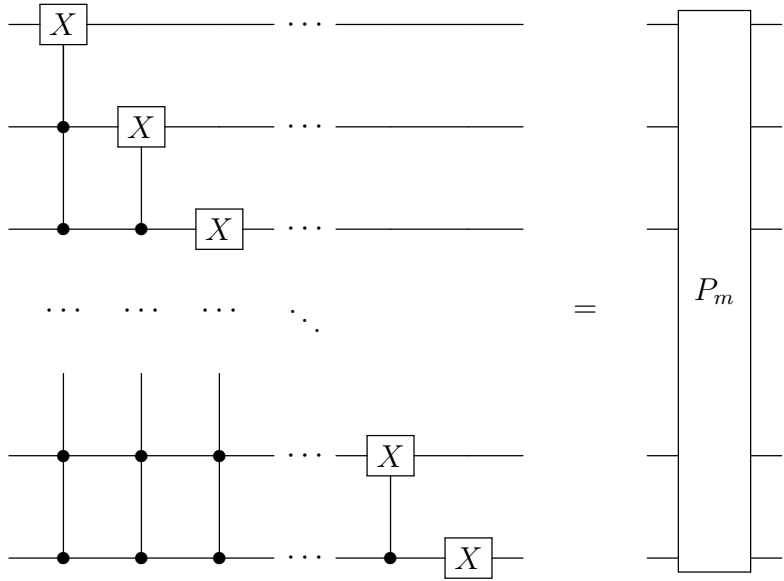

In this section we deal with the implementation of the \lq HAM-SIM\rq\  module (Figure~\ref{fig:general_circuit}) which effectively applies 
$e^{-i\hat\Delta_h t_0}$ onto register $B$. In our case the eigenvectors of the discretized Laplacian are known and we use approximations of the eigenvalues.
From (\ref{eq:Delta_h}) and (\ref{eq:A_d_dim}) we have
\begin{equation}
\label{eq:-Delta-sim}
e^{-i\Delta_h \gamma}=\underbrace{e^{i h^{-2} L_h \gamma}  \otimes \cdots \otimes e^{i h^{-2} L_h \gamma}}_{d \,\, {\rm matrices}}.
\end{equation}
Thus it suffices to implement $e^{i h^{-2} L_h \gamma}$, for certain $\gamma \in \mathbb{R}$, $\gamma
= 2\pi\cdot 2^t / E$, $t = 0,1,\ldots ,\log_2 E -1$ that are required in phase estimation.
This can be accomplished by considering the spectral decomposition $S \Lambda S$ of the matrix $L_h$, where $S$ is the matrix of the sine transform \cite{wickerhauser94,demmel97}. Then $S$ can be implemented using the quantum Fourier transform. We will implement an approximation of $\Lambda$. 

We remark that the quantum circuits presented here can be used  
in the simulation of the Hamiltonian $-\Delta + V$ using splitting formulas.
For results concerning Hamiltonian simulation using splitting formulas see \cite{berry07,zhang10,CW12}.

\subsection{One-dimensional case} \label{sec:1dcase}

We start with the implementation of $e^{i h^{-2} L_h \gamma}$, $\gamma = 2\pi 2^t / E$, 
$E = 4M^2$ when $d=1$ and $t = 0, 1, \ldots , n -1$, where $n$ is the number of qubits in 
register $C$; see (\ref{eq:prior-F}). The form of $L_h$ is shown in (\ref{eq:matrix_1d})
and is positive definite. It is a Toeplitz matrix and it is known that this type of matrices can be diagonalized via the sine transform $S$ \cite{benedetto97}. We have $L_h=S \Lambda S^\dagger$, where $\Lambda$ is an $M-1\times M-1$ diagonal matrix containing the eigenvalues 
$4 \sin^2 (j\pi /(2M))$, $j=1,\ldots ,M-1$, of $L_h$ and $S = \{S_{i,j}\}_{i,j = 1,2,\ldots ,M-1}$ is the sine transform where $S_{i,j}=\sqrt{\frac{2}{M}}$sin$(\frac{{\pi}ij}{M})$, $i,j$=1,$\dots$,$M-1$. 
Thus
\begin{equation}
\label{eq:SDDelta}
e^{ih^{-2} L_h \gamma } = S e^{i h^{-2} \Lambda \gamma} S.
\end{equation}
The relationship between the sine and cosine transforms and the Fourier transform can be found in
\cite[Thm. 3.10]{wickerhauser94}.

In particular, using the notation in~\cite{wickerhauser94}, we have
\begin{equation} \label{eq:FSC}
T_{M}^\dagger{F_{2M}}T_{M}=C_{M+1}\oplus({-iS_{M-1}}) = \left(\begin{array}{cc}
C_{M+1} & 0 \\
0 & -i S_{M-1}
\end{array} \right),
\end{equation}
where $C_{M+1}$, $S_{M-1}$ denote the cosine and sine transforms, and 
the subscripts $M-1$ and $M+1$ emphasize the size of the respective matrix. $F_{2M}$ is
the $2M \times 2M$ matrix of the Fourier transform.
The matrix $T_M$ has size $2M \times 2M$ and is given by (\ref{eq:TN}).

\begin{equation} \label{eq:TN}
T_{M}=
\left( \begin{array}{cccccccc}
1 &                    &         &                    &   &                    &         &                    \\
  & \frac{1}{\sqrt{2}} &         &                    &   & \frac{i}{\sqrt{2}} &         &                    \\
  &                    & \ddots  &                    &   &                    & \ddots  &                    \\
  &                    &         & \frac{1}{\sqrt{2}} &   &                    &         & \frac{i}{\sqrt{2}} \\
  &                    &         &                    & 1 &                    &         &                    \\
  &                    &         & \frac{1}{\sqrt{2}} &   &                    &         &-\frac{i}{\sqrt{2}} \\
  &                    & \iddots &                    &   &                    & \iddots &                    \\
  & \frac{1}{\sqrt{2}} &         &                    &   &-\frac{i}{\sqrt{2}} &         &      
\end{array} \right)
\end{equation}

The quantum circuits for implementing the unitary transformation $T_M$ is discussed in \cite{roetteler01}. The action of $T_M$ can be described by~\cite{roetteler01}

\begin{equation}\label{eq:TM_action}
\begin{array}{lll}
T_M|0x\rangle & = & \frac{1}{\sqrt{2}}|0x\rangle+\frac{1}{\sqrt{2}}|1x'\rangle \\
T_M|1x\rangle & = & \frac{i}{\sqrt{2}}|0x\rangle-\frac{i}{\sqrt{2}}|1x'\rangle
\end{array}
\end{equation}
where $i^2=-1$, $x$ is an $n$-bit number ranging $1\le{x}<2^n$ and $x'$ denotes its two's complement i.e. $x'=2^n-x$. The basic idea of implementing $T_M$ is to separate its operation into an operator $D$, which ignores the two's complement in $T_M$, and a controlled permutation $\pi$, which transforms the state $|bx\rangle$ to $|bx'\rangle$ only if $b$ is 1. Therefore the action of $D$ and $\pi$ can be written as

\begin{equation}\label{eq:D_action}
\begin{array}{lll}
D|0x\rangle & = & \frac{1}{\sqrt{2}}|0x\rangle+\frac{1}{\sqrt{2}}|1x\rangle \\
D|1x\rangle & = & \frac{i}{\sqrt{2}}|0x\rangle-\frac{i}{\sqrt{2}}|1x\rangle \\[0.2in]
\pi|0x\rangle & = & |0x\rangle \\
\pi|1x\rangle & = & |1x'\rangle
\end{array}
\end{equation}
Clearly, $T_M=D\pi$ and the overall circuit for implementing operation $T_M$ is shown in Figure~\ref{fig:TM}.

By (\ref{eq:FSC}) the sine transform $S$ can be implemented by cascading the quantum circuits in Figure \ref{fig:TM} with the circuit for Fourier transform \cite{nielsen00}. An ancilla bit is added to register $b$. It is kept in the state $|1\rangle$ in order to select the lower-right block

\begin{equation} \label{eq:sine}
{
\left( \begin{array}{cc}
a & 0 \\
0 & -iS_{M-1}
\end{array} \right)
}
\end{equation}
from the unitary operation $T_M^\dagger{F_{2M}}T_M$ (\ref{eq:FSC}), $a\in\mathbb{C}$. Considering
the state $\ket{f_h}$, that corresponds to the right hand side of (\ref{eq:matrix_1d}),
and for $b_i = \Braket{i}{f_h}$ we have 
\begin{equation} \label{eq:b}
{
({0},\underbrace{b_1,b_2,...,b_{M-1}}_{
\begin{array}{c}\scriptsize
\text{\text{ values on the}} \\ [0in]
\scriptsize \text{$(M-1)$ nodes}
\end{array}
})
}
\end{equation}
then the element $a$ in Equation \ref{eq:sine} has no effect, and the circuit in Figure \ref{fig:eAt} is equivalent to applying $(S_{M-1}{e^{ 2 \pi i \Lambda 2^t / E}}S_{M-1})$ onto the $(M-1)$ elements of $\ket{f_h}$. This is also equivalent to simulating the Hamiltonian $e^{2\pi i h^{-2} \Lambda 2^t / E}$ with the state $\ket{f_h}$ stored in register $b$. 

We implement $e^{2\pi i h^{-2} \hat \Lambda 2^t / E}$ where 
$\hat \Lambda = \{\hat \la_j\}_{j = 1,\ldots ,M-1}$ is a diagonal matrix approximating 
$\Lambda = \{\la_j\}_{j=1,\ldots ,M-1}$.

We obtain each $\hat\la_j$, $j=1,\ldots ,M-1$ by the following algorithm. The general idea is to approximate sin$x=\Im(e^{ix})=\Im((e^{ix/r})^r)$ with $W^r$ where $W=1-ix/r+x^2/r^2$ is the Taylor expansion of $e^{ix/r}$ up to the second order term. $W^r$ is computed efficiently in fixed point arithmetic using repeated squaring. The detailed steps are the following:

\vskip 1pc
\noindent{\bf Eigenvalue Simulation Algorithm (ESA):}
\begin{enumerate}
\item Let $r = 2^{\nu + 7}$ where $\nu$ is positive integer which is related to the accuracy of the result.
The inputs and the outputs of the modules below
are $s = \max\{2\nu  +9, 11 + \nu + \log_2 M\}$ bit numbers. Internally the modules may carry 
out calculations in higher precision $O(s)$, but the results are returned using $s$ bits. This 
value of $s$ follows from the error estimates in Proposition \ref{prop:ell}.
\item We perform the transformation 
\[\ket j \ket{0}^{\otimes s}\rightarrow \ket j 
\underbrace{\ket{\hat y_j = \hat x_j  / r}}_{s \,\, \rm qubits},\]
where $\hat x_j$ is the $s$ bit truncation of $x_j = \frac{\pi j}{2M}$. Note that 
$y_j=x_j/r \in (0,1)$ and $\hat y_j$ is the $s$ bit truncation of $y_j$. 
Recall that $r\geq 2$ and $2M$ are powers of $2$. Calculations are to be
performed in fixed precision arithmetic, so division does not need to be
performed actually. All one needs to do is multiply $j$ by $\pi$ 
with $O(s)$ bits of accuracy, keeping in track the position of the 
decimal point and then take the $s$ most significant bits of the result. 
\item We compute the real and imaginary parts of the complex number 
$\hat W_1$ by truncating, if necessary, the respective parts of $\hat W_0 = 1 - \hat y^2 + i \hat y$ 
to $s$ bits of accuracy; see 
(\ref{eqn:Rep_sqr}) in Proposition ~\ref{Prop: Sin-error}. This is expressed by the transformation
\[\ket{\hat y_j} \ket{0}^{\otimes s} \ket{0}^{\otimes s} \rightarrow \ket{\hat y_j} \ket{\Re (\hat W_1)} \ket{\Im (\hat W_1)}.\]
Note that since $\ket{\hat y_j}$ is $s$ qubits long, $\hat W_0$ can be computed exactly using double 
precision and ancilla qubits and the final result can be returned in $s$ qubits.

Complex numbers are implemented using two registers, holding the real and imaginary parts. Complex 
arithmetic is performed by computing the real and imaginary parts of the result.
\item We compute $\hat W_r$ approximating $\hat W_1^r$ 
using repeated squaring. Each step of this procedure is accomplished by
the transformation
\[\ket{\Re (\hat W_{2^j})}\ket{\Im (\hat W_{2^j})} \ket{0}^{\otimes s} \ket{0}^{\otimes s}
\rightarrow \ket{\Re (\hat W_{2^j})}\ket{\Im (\hat W_{2^j})} \ket{\Re (\hat W_{2^{j+1}})}\ket{\Im (\hat W_{2^{j+1}})},\]
which describes the steps in (\ref{eqn:Rep_sqr}). The registers holding real and imaginary 
parts of the numbers are $s$ qubits long.
\item $\Im (\hat W_r)$ approximates $\sin(\pi j/(2M))$ with error $2^{-(\nu - 1)}$. Hence
$\Im^2(\hat W_r)$ approximates the $\sin^2(\pi j/(2M))$.
We compute the square of $\Im (\hat W_r)$ exactly and multiply it by $4M^2$ (this 
involves only shifting). We keep the
$\nu + \log_2 (4M^2)$ most significant bits of the result, which we denote by $\ell_j$. 
This means that the $\log_2 (4M^2)$ bits of the binary string representing
$\ell_j$ compose the integer part and the last $\nu$ bits compose 
the fractional parts of the approximation to $\la_j$. Then 
\[|\la_j - \ell_j| \leq 17 \cdot 2^{-\nu} M^2.\]
For the error estimate details see Proposition \ref{prop:ell}. When $d=1$, $n$ (the
number of qubits in register $C$) and $\nu$ are related by $n = \nu + \log_2 (4M^2)$. 
Moreover, in the one dimensional case
$\hat \la_j = \ell_j$. 
\item Let $k_j$ be the binary string representing $\ell_j$. For a fixed $t$,
we implement the transformation

\begin{equation} \label{eq:transf}
\underbrace{\ket{k_j}}_{n \,\, \rm qubits} \ket{0}^{\otimes n} \rightarrow \ket{k_j} \underbrace{\ket{k_j 2^t}}_{n \,\, \rm qubits}
\end{equation}

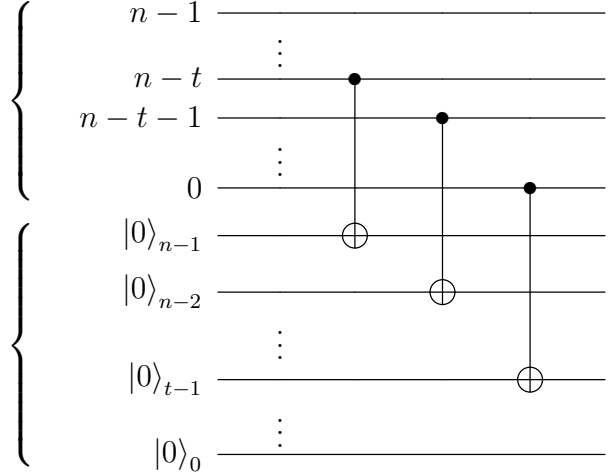
\begin{figure}
\[
\Qcircuit @R=1em @C=2em {
& & & & \lstick{n-1}     	   &   \qw   & \qw      & \qw        & \qw     	& \qw \\
& & & &                  	   & \vdots  & 		    &            & \\
& & & & \lstick{n-t}     	   &   \qw   & \ctrl{4} & \qw        & \qw		& \qw \\
& & & & \lstick{n-t-1}  	   &	\qw	 & \qw	    & \ctrl{4}   & \qw		& \qw \\
& & & &                  	   & \vdots  & 		    & 		     &				  \\
& & & & \lstick{0}       	   & \qw     & \qw	    & \qw		 & \ctrl{4}	& \qw \\
& & & & \lstick{\ket{0}_{n-1}} & \qw	 & \targ	& \qw		 & \qw		& \qw \\
& & & & \lstick{\ket{0}_{n-2}} & \qw	 & \qw	    & \targ	     & \qw		& \qw \\
& & & & 				       & \vdots  & 		    &			 &	\\
& & & & \lstick{\ket{0}_{t-1}} & \qw     & \qw	    & \qw		 & \targ	& \qw \gategroup{1}{2}{6}{2}{.7em}{\{}\\
& & & & 				 	   & \vdots  &		    &			 & 	\\
& & & & \lstick{\ket{0}_0}     &  \qw    & \qw 	    & \qw		 & \qw		& \qw \gategroup{7}{2}{12}{2}{.7em}{\{} \\
}
\]
\caption{Quantum circuit for implementing the transformation in Equation~\ref{eq:transf}.}
\label{fig:transf}
\end{figure}

This is accomplished using CNOTs with the circuit shown in Figure~\ref{fig:transf}, since $t \leq n$ the total number of quantum operations and qubits required to 
implement the circuit for all the values of $t$ is $O(n^2)$.
\item Finally, we use phase kickback (see e.g. \cite{phase-kickback}) 
to obtain $e^{2\pi i \phi_j 2^t}$ from the state $\ket{k_j 2^t}$ where $\phi_j$ is 
the phase corresponding to the eigenvalue $\ell_j$ that approximates $\la_j$; see (\ref{eq:kj}).
\end{enumerate}

\subsection{Multidimensional case}

To implement $e^{-i\Delta_h \gamma}$, $\gamma = 2\pi 2^t/E$, $E$ defined in (\ref{eq:def-E})
and $t = 0, \ldots , n-1$ we use 
\begin{equation}
\label{eq:d-dim-tensor}
e^{-i\Delta_h \gamma}=\underbrace{e^{i h^{-2} L_h \gamma}  \otimes \cdots \otimes e^{i h^{-2} L_h \gamma}}_{d \,\, {\rm matrices}}.
\end{equation}

Therefore the quantum circuit implementing $e^{-i\Delta_h \gamma}$ in $d$ dimensions
is obtained by the replication and parallel application of the circuit 
simulating $e^{i h^{-2} L_h \gamma}$. For example, when $d=2$ we have the circuit in Figure
 \ref{fig:eAt}. The register $B$ of Figure \ref{fig:general_circuit} contains $dm$ qubits, 
 $m = \log_2 M$ and its initial state is assumed
to have the form
\begin{equation} \label{eq:bd}
{
(\underbrace{0,...,0}_{M^d-(M-1)^d},\underbrace{b_1,b_2,...,b_{(M-1)^d}}_{
\begin{array}{c}
\text{\scriptsize \text{ values on the nodes of}} \\ 
\scriptsize \text{$(M-1)^{(\times{d})}$ grid}
\end{array}
})
}
\end{equation}
where $b_i = \Braket{i}{f_h}$. This way we select $S_{M-1}$ block in $T_M^\dagger F_{2M} T_M$ in 
(\ref{eq:FSC}) in each
circuit for $e^{i h^{-2} L_h \gamma}$. Recall that $\ket{f_h}$ corresponds to the right hand side
of (\ref{eq:d-dim-general}). 

The eigenvalues in the
$d$ dimensional case are given as sums of the one dimensional eigenvalues. 
We do not need to form the sums explicitly for the simulation of $-\Delta_h$; they are
computed by the tensor products. The difference between the $d$ 
dimensional and the one dimensional case is that the register $C$ in Figure \ref{fig:general_circuit} 
has $\lceil \log_2 d \rceil$ additional qubits; i.e $n = \lceil \log_2 d \rceil + \log_2 4M^2 + \nu$.
Accordingly, we generate the one dimensional approximations to the eigenvalues using the steps $1-5$
of the eigenvalue estimation algorithm of the previous section. Then we append
$\lceil \log_2 d\rceil$ qubits initialized to
$\ket{0}^{\otimes \lceil \log_2 d \rceil}$ to the left of the register holding the $\ket{\ell_j}$ 
and carry out
the remaining two steps $6-7$ with 
$n = \lceil \log_2 d \rceil + \log_2 4M^2 + \nu$. The error in the approximate eigenvalues is equal to
$17 M^2 d/ 2^\nu$; see Theorem \ref{th:d-error-eig}.

\subsection{Simulation cost}

Simulating the sine and cosine transforms (\ref{eq:FSC}) 
requires $O(m^2)$, $m = \log_2 M$ quantum operations and
$O(m)$ qubits \cite{roetteler01}. The diagonal eigenvalue matrix of the one dimensional case
(\ref{eq:SDDelta}) is simulated by
ESA. Its steps $1-3$ and step $5$ require $O(s^2)$ quantum operations
and $O(s)$ qubits. In step $4$ repeated squaring is performed $\nu + 7$ times. Each
repetition or step of the procedure requires $O(s^2)$ quantum operations and 
$O(s)$ qubits. The total cost of step $4$ is proportional to $\nu \cdot O(s^2)$ quantum
operations and $\nu \cdot O(s)$ qubits, accounting for any ancilla qubits used in 
repeated squaring.
Step $6$ requires $O(n+t)$ quantum operations and qubits for fixed $t$. Step~$7$ 
requires $O(n^2)$ quantum operations, due to the Fourier transform, and $O(n)$ qubits.

Using Theorem \ref{th:d-error-eig}, and requiring 
error $\varepsilon$ in the approximation of the eigenvalues,
we have  
\[\frac{17E}{2^\nu}\leq \varepsilon.\]

\[\nu = \lceil\log_2 \frac{17 E}{\varepsilon}  \rceil, \]
i.e. $\nu = \Theta(\log_2 d + m + \log_2 \varepsilon^{-1})$. We also have $n = \Theta (\nu)$ and $s = \Theta (n)$.

We derive the simulation cost taking the following facts into account:
\begin{itemize}
\item
Steps $1-5$ deal with the approximation of the eigenvalues. These computations are not repeated for every $t=0,\dots,n-1$.
The total cost of these steps is $O(n^3)$ quantum operations and $O(n^2)$ qubits.
\item
The total cost of step $6$, resulting from all the values of $t$, is $O(n^2)$ quantum operations and qubits.
\item
The total cost of step $7$, that applies phase kickback for all values of $t$, does not exceed $O(n^3)$ quantum operations and
$O(n^2)$ qubits.
\end{itemize}
Therefore the total cost
to simulate $e^{ih^{-2} L_h \gamma}$, $\gamma = 2\pi 2^t / E$, for all $t=0,\dots,n-1$, is $O(n^3)$ 
quantum operations and $O(n^2)$ qubits. From (\ref{eq:-Delta-sim}) we conclude that the
cost to simulate Poisson's matrix for the $d$ dimensional problem is $d\cdot O(n^3)$ 
quantum operations and $d\cdot O(n^2)$ qubits.

Finally, we remark that the dominant component of the cost is the one resulting from the approximation of
the eigenvalues (i.e., the cost of steps $1-5$). 

\section{Total cost} \label{sec:TotCost}

We now consider the total cost for solving the Poisson equation 
(\ref{eq:poisson}). Discretizing the second derivative operator
on a grid with mesh size $h = 1/M$ results to a system of linear equations, where the coefficient 
matrix is $(M-1)^d \times (M-1)^d$, i.e. exponential in the dimension $d \geq 1$. Solving this
system using classical algorithms has cost that grows at least as fast as the number of unknowns
$(M-1)^d$. For the case $d = 2$, \cite[Table 6.1]{demmel97} summarizes the cost of direct and 
iterative classical algorithms solving this system.

For simulating Poisson's matrix
we need $d \cdot O(n^3)$ quantum operations and $d \cdot O(n^2)$ qubits, where $n = O(\log_2 d + m + \log_2 \varepsilon^{-1})$
and $m = \log_2 M$. To this we add the cost for computing the reciprocal of the eigenvalues 
which is $O((\log_2 \varepsilon_0^{-1})^2 \, \log_2\log_2 \varepsilon_0^{-1})$ quantum operations
and $O((\log_2 \varepsilon_0^{-1}) \log_2 \log_2 \varepsilon_0^{-1})$ qubits,
accounting for the $O(\log_2 \log_2 \varepsilon_0^{-1})$ Newton steps, $\varepsilon_0 = \min \{\varepsilon, E^{-1}\}$. 
Finally, we add the cost of the conditional rotation which is
proportional to
$(\log_2 \varepsilon_1^{-1})^4$ quantum operations and $(\log_2 \varepsilon_1^{-1})^3$ qubits, $\varepsilon_1=\min\{ 1/(4M)^2,\varepsilon\}$.

From the above we conclude that
the quantum circuit implementing the algorithm requires of order $d \cdot O(n^3)+ (\log_2 \varepsilon_1^{-1})^4$ quantum operations
and $d \cdot O(n^2) + (\log_2 \varepsilon_1^{-1})^3$ qubits.

The relation between the matrix size and the accuracy is very important 
in assessing the performance of the quantum algorithm solving a linear system, since 
its cost depends on both of these quantities \cite{lloyd09}.
In particular, for the Poisson equation we have ignored, so far, the effect of the
discretization error of the Laplacian $\Delta$. If the grid 
is too coarse the discretization error will exceed the desired accuracy. If the
grid is too fine, the matrix will be unnecessarily large. Thus the mesh size
and, therefore, the matrix size should depend on $\varepsilon$, i.e. $M=M (\varepsilon)$. This
dependence is determined by the smoothness of the solution $u$, which, in turn, 
depends on the smoothness of the right hand side function $f$. For example,
if $f$ has uniformly bounded partial derivatives up to order four, then the 
discretization error is $O(h^2)$ and we set $M = \varepsilon^{-1/2}$; see \cite{demmel97,forth04}
for details.
In general, we have $M = \varepsilon^{-\alpha}$, where $\alpha > 0$ is a parameter
depending on the smoothness of the solution. This yields $n = O(\log_2 d + \log_2 \varepsilon^{-1})$, since
$m = \log_2 M = \alpha \log_2  \varepsilon^{-1}$. The resulting number of the 
quantum operations for the circuit is proportional to 
\[\max\{d,\log_2\varepsilon^{-1}\} (\log_2 d + \log_2 \varepsilon^{-1})^3,\]
and the number of qubits is proportional to 
\[\max\{d,\log_2\varepsilon^{-1}\} (\log_2 d + \log_2 \varepsilon^{-1})^2,\]
It can be shown that
$\log_2 d = O( \log_2 \varepsilon^{-1})$  and the number of quantum operations and qubits become proportional to 
\[\max\{d,\log_2\varepsilon^{-1}\} (\log_2 \varepsilon^{-1})^3,\]
and
\[\max\{d,\log_2\varepsilon^{-1}\} (\log_2 \varepsilon^{-1})^2,\]
respectively.

Observe that the condition number of the matrix is proportional to $\varepsilon^{-2\alpha}$ and is
independent of $d$. Therefore a number of repetitions proportional to $\varepsilon^{-4 \alpha}$
leads to 
a success probability arbitrarily close to one, regardless of the 
value of $d$. This follows 
because repeating an algorithms many times increases its probability to succeed 
at least according to the Chernoff bounds \cite[Box 3.4, pg. 154]{nielsen00}.
In contrast to this, the cost of any deterministic classical algorithm solving 
the Poisson equation is exponential in $d$. Indeed, for error $\varepsilon$ the cost is bounded from
below by a quantity proportional 
to $\varepsilon^{-d/r}$ where $r$ is a smoothness parameter \cite{wer91}.

\section{Conclusion and future directions} \label{sec:Concl}

We present a quantum algorithm and a circuit for approximating the solution of the Poisson equation in $d$ dimensions.
The algorithm breaks the curse of dimensionality and in terms of $d$
yields an exponential speedup relative to classical algorithms. The quantum circuit is scalable
and has been obtained by exploiting the structure of the Hamiltonian for the Poisson equation
to diagonalize it efficiently.
In addition, we provide quantum circuit modules for computing the reciprocal of eigenvalues
and trigonometric approximations. These modules can be used 
in other problems as well. 

The successful development of the quantum Poisson solver opens up entirely new 
horizons in solving structured systems on quantum computers, such as those 
involving Toeplitz matrices. Hamiltonian simulation 
techniques \cite{berry07,zhang10,CW12} can also be combined with our 
algorithm to extend its applicability to PDEs, signal processing, time series
analysis and other areas. 

\section*{Acknowledgements} \label{sec:Ackn}
Sabre Kais and Yudong Cao would like to thank 
NSF CCI center, ``Quantum Information for Quantum Chemistry (QIQC)", Award number CHE-1037992 and 
Army Research Office (ARO) for partial support.

Anargyros Papageorgiou, Iasonas Petras and Joseph F. Traub thank the NSF for financial support.

\section*{Appendix 1}
In this paper, $X$, $Y$ and $Z$ are Pauli matrices $\sigma_x$, $\sigma_y$ and $\sigma_z$. $I$ represents identity matrix. $H$ is the Hadamard gate and $W$ in Figure \ref{fig:general_circuit} represents $H^{\otimes{n}}$ where $n$ is the number of qubits in the register. The matrix representations of other quantum gates used are the following:

\begin{equation} \label{eq:appendix_gates2}
%I=\bigg(\begin{matrix} 1 & 0 \\ 0 & 1 \end{matrix}\bigg),\quad
V^\dagger=\frac{1}{2}
\left( \begin{array}{cc} 1-i & 1+i \\ 1+i & 1-i \end{array} \right)
,\quad
R_{zz}(\theta)=e^{i\theta}
\left( \begin{array}{cc} 1 & 0 \\ 0 & 1 \end{array} \right)
 \\
\end{equation}

\begin{equation} \label{eq:appendix_gates3}
R_x(\theta)=
\left( \begin{array}{cc} 
\text{cos}(\frac{\theta}{2}) & i\text{sin}(\frac{\theta}{2}) \\ 
i\text{sin}(\frac{\theta}{2}) & \text{cos}(\frac{\theta}{2})
\end{array} \right)
,\quad
R_y(\theta)=
\left( \begin{array}{cc} 
\text{cos}(\frac{\theta}{2}) & \text{-sin}(\frac{\theta}{2}) \\ 
\text{sin}(\frac{\theta}{2}) & \text{cos}(\frac{\theta}{2})
\end{array} \right)
\end{equation}

\begin{equation} \label{eq:appendix_gates4}
S=
\left( \begin{array}{cc} 
1 & 0 \\ 
0 & i
\end{array} \right)
,\quad
T=
\left( \begin{array}{cc} 
1 & 0 \\ 
0 & e^{i\frac{\pi}{4}}
\end{array} \right)
,\quad
R_{z}(\theta)=
\left( \begin{array}{cc} 1 & 0 \\ 0 & e^{i\theta} \end{array} \right)
\end{equation}

\section*{Appendix 2}

\begin{theorem} 
\label{th:Newton_error}
Consider the approximation $\hat x_s$ to $v^{-1}$, $v>1$, using $s$ steps of
Newton iteration, with initial approximation $\hat x_0 = 2^{-p}$, $2^{p-1} < v \leq 2^p$.
Assume that each step takes as inputs $b$ bit numbers and produces 
$b$ bit outputs and that all internal calculations are carried out in 
fixed precision arithmetic.Then the error is 
\[|\hat x_s - v^{-1}| \leq \varepsilon_N + s 2^{-b},\]
where $\varepsilon_N$ denotes the desired error of Newton iteration without considering 
the truncation error, $\varepsilon_N \geq 2^{-2^s}$. The truncation error is given by the second term 
and $s \geq \lceil \log_2\log_2 \varepsilon_N^{-1}\rceil$, $b > p$.
\end{theorem}

\begin{proof}
Consider the function $g(x) = 1/x - v$, $x>0$, where $g(1/v)=0$.
The Newton iteration for approximating the zero of $g$ is given by
\[x_{s+1} = \varphi(x_s)=2x_s - vx_s^2  \quad s=0,1,\dots .\]
The error $e_s=|x_s-1/v|$ satisfies $e_{s+1}= v e_s^2$.
Unfolding the recurrence we get
\[e_s \le (v e_0)^{2^s}.\]
Let $x_0=2^{-p}$. Now consider the least power of two that is greater than or equal to $v$, i.e., $2^{p-1}< v\le 2^p$. Clearly $p>1$ since $v>1$
and $ve_0 <1/2$. For error $\varepsilon_N$ we have $2^{-2^s} \le \varepsilon_N$, which implies $s\ge \lceil \log_2 \log_2 \varepsilon^{-1}_N \rceil$.

The derivative of the iteration function is decreasing and we have $|\varphi^\prime| \le 2(1-ax_0)\le 1$.
We will implement the iteration using fixed precision arithmetic. We first calculate the round off error.
We have 
\begin{eqnarray*}
\hat x_0 &=& x_0\\
\hat x_1 &=& \varphi(\hat x_0) + \xi_1 \\
\hat x_2 &=& \varphi(\hat x_1) + \xi_2 \\
&&\vdots \\
\hat x_s &=& \varphi(\hat x_{s-1}) + \xi_s,
\end{eqnarray*}
where the $\xi_i$ denotes truncation error at the respective steps.
Thus
\[\hat x_s - x_s = \varphi( \hat x_{s-1}) + \xi_s - \varphi(x_{s-1}),\]
and using the fact $|\varphi^\prime| \le 1$ we obtain
\[
|\hat x_s - x_s| \le | \hat x_{s-1} -x_{s-1} | + |\xi_s|  \le \sum_{i=1}^s |\xi_i|\le s 2^{-b},
\]
assuming that we truncate the intermediate results to $b$ bits of accuracy.
\end{proof}

\begin{lem}
\label{lem1:Sin_error}
Let $x \in [\pi/(2M),\pi/2)$ and $W = 1+i \frac{x}{r} - \frac{x^2}{r^2}$. Then 
\[ \left| e^{ix} - W^r \right| \leq 2^7/r.\]
\end{lem}

\begin{proof} $e^{ix} = \left( e^{ix/r}\right)^r = \left( W+E(x/r)\right)^r$, where for
$y = x/r$, $E(y) = \sum_{k \ge{3}} \frac{(iy)^k}{k!} $ and
\begin{equation} \label{eq:Ey}
\begin{tabular}{l}
$\displaystyle
\left|\sum_{k\ge{3}}{\frac{(iy)^k}{k!}}\right|\le\sum_{k\ge{3}}{\frac{|y|^k}{k!}}=|y|^3\sum_{k\ge{3}}^{}{\frac{|y|^{k-3}}{k!}}
=|y|^3\sum_{k\ge{0}}{
\underbrace{\frac{k!}{(k+3)!}}_{\text{$\frac{1}{(k+1)(k+2)(k+3)}\le\frac{1}{6}$}}
\frac{|y|^k}{k!}}$ \\
$\displaystyle
\le\frac{|y|^3}{6}e^{|y|}<|y|^3$
\end{tabular}
\end{equation}
where the last inequality holds for $|y|=|\frac{x}{r}|<1$, which is true due to our assumptions. Hence $\left|E(\frac{x}{r})\right|\le|\frac{x}{r}|^3$ for $|x|<r$. 

We then turn our attention to the powers of $W$.
\begin{equation} \label{eq:W}
\begin{tabular}{c}
$\displaystyle
|W|=|1+i\frac{x}{r}-\frac{x^2}{r^2}|\le{1+\frac{x}{r}+\frac{x^2}{r^2}}$ \\
\end{tabular}
\end{equation}
For all $k\in\{1,2,...,r\}$ we have,
\begin{equation} \label{eq:Weq}
\begin{tabular}{c}
$\displaystyle
|W|^k\le \left(1+\frac{x}{r}+\frac{x^2}{r^2} \right)^k\le{e^{\left(\frac{x}{r}+\frac{x^2}{r^2}\right)k}}
=e^{\frac{|x|}{r}k}e^{\frac{|x|^2}{r^2}k}\le{e^{|x|}}e^{\frac{|x|^2}{r}} \leq e^{2x} \leq e^\pi$.
\end{tabular}
\end{equation}
where we have used the fact that $\frac{k}{r}<1$. The second inequality is due to $(1+a)^k\le{e^{ka}}$,
$a\in\mathbb{R}$, $k\in\mathbb{Z}^+$. Indeed

\begin{equation} \label{eq:akineq}
\begin{tabular}{c}
$\displaystyle
(1+a)^k=\sum_{l=0}^k
\left( \begin{array}{c} k \\ l \end{array} \right)
% \begin{pmatrix} k \\ l \end{pmatrix}
a^{k-l}=
\sum_{l=0}^k\frac{k!}{l!(k-l)!}a^{k-l}=\sum_{l=0}^k\frac{k!}{l!(k-l)!}\frac{(ka)^{k-l}}{k^{k-l}}=$ \\
$\displaystyle
\sum_{l=0}^k\underbrace{\frac{k(k-1)\cdots (l+1)}{k^{k-l}}}_{\text{$\le{1}$}}
\underbrace{\frac{l\cdots 1}{l!}}_{\text{$\le{1}$}}
\frac{(ka)^{k-l}}{(k-l)!}\le\sum_{l=0}^k\frac{(ka)^{k-l}}{(k-l)!}
=\sum_{l=0}^k\frac{(ka)^l}{l!}\le{e^{ka}}$
\end{tabular}
\end{equation}
Finally we look at the approximation error. Note that

\begin{equation} \label{eq:approx_err}
\begin{array}{ccl}
e^{ix} &=& \displaystyle\left(W+E\left(\frac{x}{r}\right)\right)^r=\sum_{k=0}^{r}
\left( \begin{array}{c} r \\ k \end{array} \right)
W^k\left[E\left(\frac{x}{r}\right)\right]^{r-k} \nonumber \\
&=& W^r+  \underbrace{
\left( \begin{array}{c} r \\ l \end{array} \right)
W^{r-1}E\left(\frac{x}{r}\right)+ \ldots +
\left( \begin{array}{c} r \\ r \end{array} \right)
W^{0}\left[E\left(\frac{x}{r}\right)\right]^r}_{\text{error in $r$-th power}}
\end{array}
\end{equation}

Consider the $k$-th term in the error series. According to (\ref{eq:Ey}) we have
\begin{eqnarray*} \label{eq:error_term}
\left( \begin{array}{c} r \\ k \end{array} \right)
|W|^{r-k} \left|E\left(\frac{x}{r}\right)\right|^{k} &\le & C  
\left( \begin{array}{c} r \\ k \end{array} \right)
|\frac{x}{r}|^{3k}=C\frac{r!}{k!(r-k)!}\frac{|x|^{3k}}{r^{3k}}\\
& = &C\frac{r(r-1)\cdots (r-k+1)}{k!}\frac{1}{r^k}\frac{|x|^{3k}}{r^{2k}} \\ 
&\le & 
C\frac{|x|^k}{k!}\frac{|x|^{2k}}{r^{2k}}\le \frac{\pi}{2} C \left(\frac{|x|}{r}\right)^{2k} \leq \frac{\pi}{2} e^\pi \left(\frac{|x|}{r}\right)^{2k}, 
\end{eqnarray*}
where $C= e^\pi$ and we use Stirling's formula $k ! = \sqrt{2\pi} k^{k+1/2} \exp \left( -k + \frac{\theta}{12 k }\right)$, 
$\theta \in (0,1)$, \cite[p. 257]{abra} to obtain $|x|^k / k! \leq 5^{-k} x^k e^k \leq 1$ for $k \geq 5$, since
$|x|\le\frac{\pi}{2}$. So the total approximation error is bounded by
\begin{equation} \label{eq:error_bound}
|e^{ix}-W^r|\le\sum_{k=1}^r
\left(
\begin{array}{c} r \\ k \end{array}
\right)
|W|^{r-k}\left|\frac{x}{r}\right|^{3k}\le
\frac{\pi}{2} e^\pi r(\frac{|x|}{r})^2\le e^\pi \left(\frac{\pi}{2}\right)^3
\frac{1}{r} \leq 2^7 \cdot \frac{1}{r}
\end{equation}
\end{proof}

\begin{lem}
\label{lem:cos-sin} 
Under the assumptions of Lemma \ref{lem1:Sin_error}
\[\left|\sin x - \Im \left(W^r \right) \right|\leq 2^7 / r\]
and \[\left| \cos x - \Re \left(W^r\right) \right| \leq 2^7 / r.\]
\end{lem}

The proof is trivial and we omit it.

\begin{prop}
\label{Prop: Sin-error}
Let $r = 2^{\nu +7}$ for $\nu \geq 1$ and consider the procedure computing
$W^r$, as defined in Lemma \ref{lem1:Sin_error} using repeated squaring. Assume each 
step computing a square carries out the calculation using fixed precision arithmetic
and that its inputs and outputs are $s$ bit numbers. Let $\hat W_r$ be the
final result. Then the error is
\[\left| W^r - \hat W_r \right| \leq \frac{2^{\nu +9}}{2^s},\]
for $s \geq 11 + \nu + \log_2 M$, where $1/M$ is the mesh size in the discretization of the
Poisson equation.
\end{prop}

\begin{proof}
We are interested in estimating $\sin (j\pi/(2M))$, for $j= 1,2,\ldots ,M-1$. We consider $x \in [\pi/(2M) , \pi/2)$. We
approximate $e^{ix}$ and from this $\sin x$, which is the imaginary part of $e^{ix}$. Let $y = \frac{x}{r}\leq 2^{-7}$. We truncate it 
to $s$ bits of accuracy to obtain $\hat y$. Note that $W = 1- y^2 + i y$ satisfies $|W|^2 = 1 - y^2 + y^4 < 1$. Let $\hat W_0 
= 1 - \hat y^2 + i \hat y$, $y - \hat y \leq 2^{-s}$. Then $|\hat W_0|^2 \leq |W|^2 + 4 y 2^{-s} < 1$, for 
$s \geq 11 + \nu + \log_2 M$.
This value of $s$ follows by solving 
\[4 y2^{-s} \leq y^2 / 2,\]
which ensures that $\hat W_0^2 \leq 1$. In addition 
\[\left|\Re \left(\hat W_0 - W\right)\right| \leq 2 y 2^{-s} + 2^{-2s}\]
and 
\[\left| \Im \left(\hat W_0 - W\right) \right|\leq 2^{-s}.\]

Define the sequence of approximations
\begin{eqnarray}
\label{eqn:Rep_sqr}
\hat W_1 &=& \hat W_0 + e_1 \nonumber \\
\hat W_2 &=& \hat W_1^2 + e_2 \nonumber \\
& \vdots & \nonumber \\
\hat W_r &=& \left( \hat W_{r/2}\right)^2 + e_r, 
\end{eqnarray}
where $r = 2^{\nu+7}$ and the error terms $e_1, e_2, \ldots , e_r$ are complex numbers 
denoting that the real and imaginary parts of the results are truncated 
to $s$ bits of accuracy.

Observe that if $|\hat W_{2^{j-1}}| < 1$ then 
$|\hat W_{2^{j}}| < 1$, since $|\hat W_{2^{j-1}}|^2 < 1$
and truncation of real and imaginary parts does not increase the magnitude of a complex number.
Since $|\hat W_0| < 1$, all the numbers in the sequence (\ref{eqn:Rep_sqr}) belong to the unit
disk $S$ in the complex plane. 

Let $z = a + b i$. Then the function that computes $z^2$ can be understood as 
a vector valued function of $2$ variables, $h : S \rightarrow S$, such that
$h(a,b) = (a^2 - b^2, 2ab)$. The Jacobian of $h$ is 
\[J = 2 \left(\begin{array}{cc}
a & -b \\
b & a 
\end{array}\right) \quad (a,b) \in S\]
and its Euclidean norm satisfies $\|J\| \leq 2$, since $a^2 + b^2 \leq 1$.
Using this bound we obtain 
\begin{eqnarray}
\label{eq:final_error1}
|W^r - \hat W_r | &\leq & |W^r - (\hat W_{r/2})^2| + |e_r| \nonumber \\
&\leq & 2 \{ 2 |W^{r/4} - \hat W_{r/4}| + |e_{r/4}|\} + |e_r| \nonumber \\
&\leq & 2^{\nu +7} |W - \hat W_1| + 2^{\nu +7-1} |e_2| + \ldots + 2^0 |e_{2^{\nu+7}}| \nonumber \\
&=& 2^{\nu +7} \left| W-\hat W_0 \right| + 2^{\nu +7} |e_1| + \ldots + |e_{2^{\nu +7}}| \nonumber \\
&\leq & 2^{\nu +7} \left| W-\hat W_0 \right| + \frac{\sqrt 2}{2^s} \sum_{j=0}^{\nu +7} 2^{\nu + 7 -j} 
\nonumber \\
&\leq & 2^{\nu +7} \sqrt{\left( 2 y \frac{1}{2^s} + \frac{1}{2^{2s}}\right)^2 +\frac{1}{2^{2s}}} + \frac{\sqrt 2}{2^s} \left(2^{\nu +8} -1\right) \nonumber \\
&\leq & 4 \frac{2^{\nu +7}}{2^s},
\end{eqnarray}
where the last inequality follows since $2 y + 2^{-s} \leq 2^{-6} + 2^{-11}$.
\end{proof}

\begin{prop}
\label{prop:ell}
Under the assumptions of Proposition \ref{Prop: Sin-error}
we approximate $\sin x$ by $\Im (\hat W_r)$, $x \in [\pi /(2M), \pi/2 )$, 
with $s = \max\{2\nu +9, 11 + \nu + \log_2 M\}$ bits and $r = 2^{\nu + 7}$. Then
the error is
\[|\sin x - \Im (\hat W_r)|\leq 2^{-(\nu -1)}.\]
Moreover:
\begin{itemize}
\item Denoting by $\hat W_{r,j}$ the approximations to
$\sin (\pi j/(2M))$, $j = 1,2,\ldots ,M-1$, we have the
following error bound 
\[ \left| 4 M^2 \sin^2 (j \pi /(2M)) - 4 M^2 \left(\Im (\hat W_{r,j})\right)^2 \right| 
\leq 2^{-(\nu -4)} M^2,\] 
$j = 1,2,\ldots ,M-1$, for the eigenvalues of the matrix $h^{-2} L_h$ that
approximates the second derivative operator, using mesh size $h = 1/M$.
\item Letting $\ell_j$ be the truncation of $4 M^2 \left(\Im (\hat W_{r,j})\right)^2$ to
$\nu$ bits after the decimal point (the length of $\ell_j$ is $\nu + \log_2 (4M^2)$ bits,
and $\nu$ is sufficiently large to satisfy the accuracy requirements) we have 
\[\left| 4 M^2 \sin^2 (j \pi /(2M)) - \ell_j\right| \leq 17\cdot 2^{-\nu}M^2,\]
for $j = 1,2,\ldots ,M-1$.
\end{itemize}
\end{prop}

\begin{proof} We have 
\begin{eqnarray}
\left| e^{ix} - \hat W_r \right| &\leq & \left|e^{ix} - W^r\right| + 
\left| W^r - \hat W_r \right| \nonumber \\ 
&\leq & \frac{2^7}{2^{\nu +7}} + \frac{2^{\nu+9}}{2^s} \nonumber \\ 
&=& 2^{-\nu} + \frac{2^{\nu +9}}{2^s} = \frac{1}{2^{\nu -1}},
\end{eqnarray}
for $s = \max\{2\nu  +9, 11 + \nu + \log_2 M\}$, which completes the proof of the first part.
The proof of the second and third part follows immediately. 
\end{proof}

\begin{theorem}
\label{th:d-error-eig}
Consider the eigenvalues 
\[\la_{j_1, \ldots ,j_d}  = 4M^2 \prod_{k=1}^d \sin^2 \left(\frac{j_k \pi}{2M}\right),\]
$j_k = 1,2,\ldots ,M-1$, $k=1,2,\ldots ,d$ of $-\Delta_h$, $h=1/M$. Let 
\[\hat \la_{j_1 , \ldots , j_d} = \sum_{k=1}^d \ell_{j_k},\]
where $\ell_{j_k}$ are defined in Proposition
\ref{prop:ell}, $j_k = 1,2,\ldots ,M-1$, $k=1,2,\ldots ,d$.
Then 
\[|\la_{j_1 , \ldots j_d} - \hat \la_{j_1, \ldots , j_d}| \leq \frac{17M^2 d}{2^{\nu}}.\]
\end{theorem}

The proof follows from Proposition \ref{prop:ell} and the fact that the $d$ dimensional eigenvalues 
are sums of the one dimensional eigenvalues.
$\quad$\\
\bibliographystyle{unsrt}
\bibliography{paper}

\end{document}